\documentclass[a4paper,11pt]{article}
\usepackage[english]{babel} 
\usepackage[T1]{fontenc} 
\usepackage{lmodern,microtype} 
\usepackage{eurosym}
\usepackage[utf8]{inputenc}
\usepackage{graphics,graphicx}
\usepackage[flushleft]{threeparttable}
\usepackage{lscape,float,caption}
\usepackage{multirow}
\usepackage{extsizes}

\usepackage{amsmath,amsthm,amssymb,amsfonts} 
\usepackage{dsfont,mathrsfs,ushort} 
\usepackage{braket}
\allowdisplaybreaks

\newtheorem{prediction}{Prediction}
\newtheorem{assumption}{Assumption}
\newtheorem{proposition}{Proposition}
\newtheorem{lemma}{Lemma}
\newtheorem{corollary}{Corollary}

\usepackage{titlesec,titling} 
\usepackage[nohead]{geometry} 
\usepackage{setspace} 
\usepackage{enumitem,booktabs} 
\usepackage[dvipsnames]{xcolor}
\usepackage[semicolon]{natbib} 
\usepackage{pstricks} 
\usepackage{hyperref}

\usepackage{authblk}
\long\def\comment#1{}

\title{Revisiting Granular Models of Firm Growth}
\author[1,2,3]{Jos\'{e} Moran}
\affil[1]{Macrocosm Inc., Brooklyn NY}
\affil[2]{Institute for New Economic Thinking at the Oxford Martin School, University of Oxford}
\affil[3]{Complexity Science Hub, Vienna}
\author[4]{Angelo Secchi}
\affil[4]{PSE - Universit\'e Paris 1 Panth\'eon-Sorbonne}
\author[5,6,7]{Jean-Philippe Bouchaud}
\affil[5]{Capital Fund Management, Paris}
\affil[6]{Chair of Econophysics and Complex Systems, \'{E}cole Polytechnique}
\affil[7]{Académie des Sciences, Paris}

\setenumerate{label=\small(\roman*)}
\hypersetup{colorlinks=true,pdfnewwindow=true,pdfstartview=FitH,%
pdftitle="",pdfauthor="",%
urlcolor=black,citecolor=blue!90!red!45!black,linkcolor=red!90!black}
\titleformat{\section}[block]{\flushleft\large\bfseries}{\thesection.}{0.5em}{}
\titleformat{\subsection}[block]{\flushleft\bfseries}{\thesubsection.}{0.5em}{}
\titleformat{\subsubsection}[runin]{\normalsize\itshape}{\bfseries\thesubsubsection.}{0.5em}{}[.--\:]
\renewcommand{\thesubsubsection}{\arabic{section}.\arabic{subsection}.\alph{subsubsection}}
\titlespacing{\section}{0ex}{6ex}{3ex}
\titlespacing{\subsection}{0in}{3ex}{1.5ex}
\titlespacing{\subsubsection}{0mm}{2ex}{0.5em}
\renewcommand{\linespread}[1]{\setstretch{1}}

\usepackage{tabularx}
\usepackage{colortbl,dcolumn}
\definecolor{gp2}{HTML}{318CE7}
\definecolor{gp1}{HTML}{318CE7}
\usepackage{soul}

\newcommand{\dint}{\mathrm{d}}

\newcommand{\be}{\begin{equation}} \newcommand{\ee}{\end{equation}}

\newcommand{\ben}{\begin{enumerate}} \newcommand{\een}{\end{enumerate}}
\newcommand{\bc}{\begin{center}} \newcommand{\ec}{\end{center}}
\newcommand{\bi}{\begin{itemize}} \newcommand{\ei}{\end{itemize}}
\newcommand{\iu}{\mathrm{i}\mkern1mu}

\begin{document}

\maketitle

\begin{abstract}

  We revisit granular models that represent the size of a firm as the sum of the sizes of multiple constituents or sub-units. Originally developed to address the unexpectedly slow reduction in volatility as firm size increases, these models also explain the shape of the distribution of firm growth rates.

We introduce new theoretical insights regarding the relationship between firm size and growth rate statistics within this framework, directly linking the growth statistics of a firm to how diversified it is. The non-intuitive nature of our results arises from the fat-tailed distributions of the size and the number of sub-units, which suggest the categorization of firms into three distinct diversification types:  well-diversified firms with sizes evenly distributed across many sub-units, firms with many sub-units but concentrated size in just a few, and poorly diversified firms consisting of only a small number of sub-units.

Inspired by our theoretical findings, we identify new empirical patterns in firm growth.  Our findings show that growth volatility, when adjusted by average size-conditioned volatility, has a size-independent distribution, but with a tail that is much too thin to be in agreement with the predictions of granular models. Furthermore, the predicted Gaussian distribution of growth rates, even when rescaled for firm-specific volatility, remains fat-tailed across all sizes. Such discrepancies not only challenge the granularity hypothesis but also underscore the need for deeper exploration into the mechanisms driving firm growth.
\end{abstract}

\newpage

\section{Introduction} 

What statistical laws does the growth of firms abide to? The growth
dynamics of firms depends of course on a host of both common and
idiosyncratic time-dependent factors and hence a purely statistical
approach to this question might appear somewhat futile. Still the research of the past three decades has made it
increasingly evident that changes in macro aggregates like GDP are
best understood in light of the statistics of business activities at
the firm level~\citep{haltiwanger_1997}. Moreover, the questions raised by
company growth statistics are deeply related not only to macroeconomic
fluctuations, but also to the growth of individual wealth
or the growth of city sizes, topics that have also received a large
amount of attention in an interdisciplinary
literature~\citep{Gabaix2009}. Further to this, and in spite of emerging from very heterogeneous statistical samples, many statistical facts pertaining to firm growth have been shown to be robust and to some extent universal, in that they hold across different countries, time periods and levels of
aggregation, and in that they do so independently of the size proxy
that is used \citep{buldyrev2020rise,dosi2023foundations}.

Among the best known statistical properties of firm growth is the fact
that both the firm size distribution and the growth rate distribution are
non-Gaussian \citep{axtell2001,stanley1996scaling} and feature heavy tails. As a forceful
illustration of the latter, we represent in Figure~\ref{fig:motiv-evid} (left-panel,
purple line) the empirical density of year-on-year sales growth rates
computed over a large sample of publicly traded companies in the
United States.\footnote{This figure uses data on publicly traded US
  companies drawn from COMPUSTAT. Details on the data are reported in
  Section~\ref{sec:data} and in the Notes under the figure.} As is
well documented in the literature, the distribution is to a large
extent symmetric, which in itself is somewhat unexpected, and features
a ``cusp'' in the central part along with slowly decaying tails. This
implies relatively frequent episodes of extreme growth or extreme
decline: the probability to observe sales multiplied or divided by
more than 150 is as large as one in one thousand, an event nearly
impossible if the growth rates follow a Gaussian distribution.

However, characterizing this distribution when the growth rates of
different firms are pooled together is clearly not a neutral
choice. Implicitly, it assumes that all the observations are drawn
from the same, firm-independent and time-independent
distribution. This is clearly unreasonable, in particular with regard
to another well-known statistical property, first established by
\cite{hymer_pashigian_1962}, suggesting a decline of growth
volatility with size that looks inconsistent with a simple
diversification argument. In the mid-nineties, thanks to the
availability of comprehensive balance-sheet data,
\cite{stanley1996scaling} and \cite{amaral1997scaling} provided a
quantitatively accurate characterization of this scaling relation. The
standard deviation of growth rates is shown to decrease with a firm's
size $S_i$ as a power-law $S_i^{-\beta}$ where $\beta$ is found
approximately equal to $0.2$.\footnote{Interestingly, and in spite of
  being a quantity that does not \textit{a priori} follow the same
  dynamics as firm sales, one finds the same dependence of the
  volatility of the returns of traded stocks as a function of their
  market capitalisation. Such a coincidence was also noted by \cite{herskovic2020firm}.} This decay is significantly slower than the one a simple diversification argument would suggest, which would correspond to $\beta=0.5$.

This paper addresses the question of explaining the emergence of these two empirical
laws. We consider a concise market economy comprising firms with
exogenous production and operating in independent markets. Because of
the lack of direct competition among firms, these models are known as
``island models'' \citep{Sutton_1997}. In this framework, reminiscent
of the one introduced by~\cite{gibrat1931inegalites} to explain the
firm size distribution, a firm's size evolves according to a
multiplicative process, modeling its growth rate as a random
variable. 

The simplest idea to explain the two empirical laws under scrutiny
would then be to assume that idiosyncratic growth shocks are Gaussian,
but with a firm-dependent volatility $\sigma_i$ that only depends on
the size $S_i$ of firm $i$ as in $\sigma_i = \sigma_0
S_i^{-\beta}$. This is equivalent to saying that a firm's growth
volatility is uniquely determined by its size, and imposing a
black-box type relation on the many reasons explaining why larger
firms, possibly because of diversification, have less volatile growth. The resulting distribution of growth rates can be seen as a \textit{Gaussian
  mixture}, as it is the superposition of Gaussian distributions with
heterogeneous variances. The family of ``scale mixtures of Gaussian
distributions'' is known to have generically fatter tails than a
Gaussian distribution, and it is versatile enough to generate a wide
class of continuous, unimodal and symmetric
distributions~\citep{west1987,andrewsmallows1974}.\footnote{Among the
  distributions that can be represented as Gaussian mixtures, one
  finds the Student-$t$, Logistic, Laplace, L\'evy-stable and
  Power-exponential distributions, all of which have been considered
  as good candidates to describe the empirical distribution of growth
  rates.} 
  
Unfortunately, as is illustrated in Figure~\ref{fig:motiv-evid}, this simple reasoning still falls short of explaining the empirical distribution of growth rates: although the resulting Gaussian mixture is indeed heavy tailed (dark-green line), as seen from comparing it with the Gaussian benchmark (black line), it misses its empirical target (purple line). In addition to this, the right panel of the same figure shows that although the volatility-size scaling relation holds in a wide range, it does not describe the volatility of very large corporations very well.\footnote{Evidence pointing in the same direction has been found for French firms in \cite{fonase2024}.} A more careful analysis of the data also suggests important fluctuations around the average value of growth volatility conditional on size.

To overcome this problem, the literature has considered the possibility that
firms with similar size may have additional sources of variability of
their growth trajectories. A simple way to obtain this result is to
enrich basic island models with two intertwined features. The first is that
firms can operate in more than one market. As a consequence, each firm
can be seen as an ensemble of sub-units, with each functioning in one
of these distinct and independent markets. The second is that these
sub-units may not necessarily be of equal size, so that a firm's
aggregate size can be concentrated in only a ``granular'' handful of
them. \cite{Sutton_2002}, \cite{Fu_et_al_2005},
\cite{Buldyrev_et_al_2007} and \cite{schwarzkopf_et_al_2010} provide
examples of growth models adopting this strategy. While these
\textit{compositional models} all share the feature that the
distribution of growth rates they generate is a Gaussian mixture that
can reproduce the shape of the empirical distribution quite well they
fail in capturing the size-volatility relation
\citep{moranriccaboni}.\footnote{Gaussian mixtures may emerge also in non-compositional models.
\cite{bose2006rand} presents an example where a firm's growth
rate is defined as a sum of random variables whose number is itself
random. Under proper assumptions, pooling together different firms induce a Laplace asymptotic distribution of their growth rates."} Among the same family, two notable exceptions
are presented in \cite{wyart2003statistical} and
\cite{Gabaix_2011}.\footnote{Another exception would be
  \cite{Sutton_2002}. The model presented therein leads to an average
  growth volatility conditional on size that scales with an exponent
  equal to $1/4$, close to the empirical value 0.2. However, the model
  is developed assuming that all partitions of a firm's size into
  sub-units are equiprobable (microcanonanical hypothesis), an
  hypothesis difficult to justify in our context, see also the discussion in \cite{wyart2003statistical}.}  Their two models
are indeed able to accommodate the slow decay of growth volatilities
with size and the non-Gaussian distribution of growth rates allowing
for the so-called ``granular'' effects. For this reason, we name them
``granular models of firm growth'', and they represent the key
antecedents with respect to which the present work makes its main
contributions.\footnote{For a version of the model where ``granularity'' comes from supplier-customer network effects, see \cite{herskovic2020firm}.}

\begin{figure}[t]
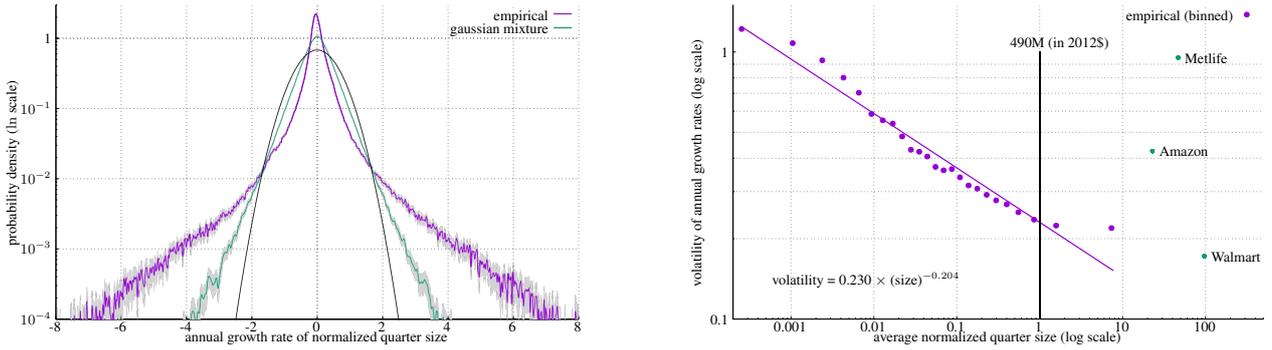

  \caption{Distribution of the annual growth rate and its size dependent volatility}
  \label{fig:motiv-evid}
  \begin{center}
    \begin{minipage}[t]{0.47\linewidth}
      \scalebox{0.6}{
        \input{./figures/fig_gr_norm_centered_empirical.tex}
      }
    \end{minipage} 
    \hfill
    \begin{minipage}[t]{0.47\linewidth}
      \scalebox{0.6}{
        \input{./figures/fig_mad_norm_scaling_firm.tex}
      }
    \end{minipage}
  \end{center}
  {\scriptsize {\it Notes}: \underline{Left-panel} reports the kernel
    estimate of the unconditional empirical density of the growth
    rates after removing the mean (purple line). The estimate is
    computed using the entire sample pooling annual (log) growth rates
    across firms and year-quarter with a Gaussian kernel function and
    a bandwidth chosen according to the normal reference
    rule-of-thumb. The estimate is evaluated on a 10,000 points
    regular grid defined over the entire empirical range. We also
    report the 95\% confidence interval and the Gaussian parametric
    fit (dotted line). We add the empirical density of synthetic
    growth rates (green line) obtained by mixing Gaussian random
    shocks with growth volatilities bootstrapped from those estimated
    according to $\sigma(S) \propto S^{-\beta}$, as suggested by the right panel, and a Normal fit
    (black line). \underline{Right-panel} displays on a double
    log-scale the binned relation between firms growth volatility and
    average normalized size. Each point represents the average size
    and the average volatility of the growth rates of firms belonging
    to the same size bin where volatility is proxied with the mean
    absolute deviation adjusted by the factor $\sqrt{\pi/2}$. We also
    report a power-law fit $\sigma(S)\propto S^{-\beta}$ where the
    exponent $\beta$ is found approximately equal to $-0.204$ with an
    asymptotic standard error of $0.01$. Note however that the
    rightmost part of the data decays even more slowly. Data source:
    Compustat. \par}
\end{figure}

We propose a unified theoretical framework in which \cite{Gabaix_2011}
and \cite{wyart2003statistical} only differ through the assumptions
they make on the statistics of the number of sub-units composing each firm: a number that is roughly proportional to the firm's size in the first case, and a random variable with a skewed
distribution for the latter. We refine the analyses of the two models
and provide formal proofs to some of the conjectures they raised,
while also correcting some of their results. In particular, we show
that growth volatility is not only driven by firm size but by a more
complicated interplay between a firm's size and its structure. This is
explained by showing that these models imply the existence of three
different types of firms: those with sales evenly distributed among a
large number of sub-units, those with a large number of sub-units but
with their sales concentrated only on a few of them, and lastly those
made up of only a handful of sub-units. Such a description, in spite
of being well adapted to the description of large firms, paints very
different pictures of how well diversified they are, an issue also
raised in the international trade literature by~\cite{Kramarz2020}. We
show that the co-existence of these three types of firms has a direct
impact on the shape of the growth volatility distribution and we
derive non-trivial results on the scaling laws of its first four
moments. This allows us to discuss in details the distribution of the
growth rates predicted by the model since in this framework it simply
results from mixing Gaussian random variables using the volatility
distribution as a mixing function.

In addition to these theoretical results, we provide a twofold empirical contribution. Our first contribution is the determination of new
empirical facts about the growth dynamics of business firms. Guided by
the theoretical considerations above, we show that the distribution of
the growth rate volatility, once rescaled by its average
conditional on size, is to a good approximation independent of firm
size, and is well approximated by an Inverse Gamma distribution, which
features a power-law right tail. We also document that the first four
moments of growth volatility, and not only the average, scale down
with size as a power-law. Our second contribution is to show that
although these observations are to a first approximation in good
qualitative agreement with granular models of growth, a deeper
investigation unveils three major inconsistencies. The first of these
is that the tail of the empirical distribution of the growth
volatility is much too thin to be compatible with the tail a granular model would predict. This means that we do not
observe large volatilities in the data as often as predicted by a calibrated granular model. The second empirical inconsistency is that moments higher than order one do not all decrease with size at the same scaling rate, at odds with what the model predicts.  Finally, the last inconsistency is that the two
granular models we have analyzed imply that a firm's growth rate,
when rescaled by the firm-dependent volatility, should be distributed
as a Gaussian random variable. We expect this to be true in particular
for large firms: since in the two models their growth rates are the sum
of a large number of (possibly non-Gaussian) independent random variables,
the Central Limit Theorem predicts that they should be partially
Gaussian. However, their empirical distribution deviates from a
Gaussian benchmark, and we quantify this deviation by proposing a
generic family of distributions that features a Gaussian central
region of varying width and exponentially ``stretched''
tails. Somewhat surprisingly, parametric fits of these distributions show that the tails do not
gradually disappear when considering larger firms, even though they do
(at a slow rate) when considering longer time horizons in building
growth rates.

Putting it all together, our analysis shows that compositional growth
models, and in particular the granular growth models analyzed in this
paper, are not able to account for the rich statistical features
characterizing the growth of business firms. This ultimately suggests
that the overarching mechanisms driving their evolution are not yet
satisfactorily understood. Possible mechanisms to alleviate this are
suggested in the concluding section.

\section{Modeling firm growth: a rigorous statistical approach}
\label{sec:theory}

In this section, we present a statistical framework to describe the
dynamics of firms' growth and how they are modulated by their
size. Our primary objective is to investigate the extent to which this
framework is compatible with the two empirical facts motivating this
paper, namely (i) the non-Gaussian shape of the growth rates
distribution and (ii) the slow decay with size of the volatility of
firms growth rates. After defining the basic set-up, we consider two
instances of the statistical framework featuring one and two sources
of granularity so revisiting the models presented in
\cite{Gabaix_2011} and \cite{wyart2003statistical} respectively. We
refine their analyses and correct some of their results, deriving new
and more precise predictions that can be directly compared with data.

\subsection{Set-up}
\label{sec:set-up}

We analyze a concise market economy comprising $N$ firms, which are
characterized by an exogenous production. 
We posit that each sufficiently large firm
consists of a specific number of sub-units, which can be construed as
departments or production units, each of which functions within
distinct and independent sub-markets. These units might alternatively thought of as representing lines of business with different customers, in the spirit of \citet{herskovic2020firm}.

The size of firm $i$ at time
$t$ is denoted by $S_{it}$ and satisfies
$S_{it}=\sum_{j=1}^{K_i} s_{jit}$, where $K_i$ represents the number
of sub-units and $s_{jit}$, $j=1,\ldots,K_i$, denotes their respective
sizes. The evolution of each sub-unit's size over time is assumed to
follow a multiplicative process with growth rates of finite variance. This definition becomes the following assumption.
\begin{assumption}
  \label{hp:gibrat}
  The time evolution of each sub-unit's size is characterized by a
  multiplicative process
  \begin{equation}
    \label{eq:sub-growth}
    s_{jit+1}=(1+\sigma_{0}\eta_{jit}) s_{jit}\;\;\;,
  \end{equation}
  where $\eta_{jit}$ are independent random shocks with zero mean
  and unit variance. The time invariant parameter $\sigma_0$ defines
  the order of magnitude of the growth fluctuations at the level
  of a sub-unit. For definiteness, the unit of $t$ is given in
  years.
\end{assumption}
Assumption~\ref{hp:gibrat} implies that the absolute growth of each
sub-unit is proportional to its size, something known in the
litterature as the ``Law of Proportionate Effects'' or ``Gibrat's
Law''. The assumption that all the sub-unit shocks $\eta_{ji}$ have a standard
deviation $\sigma_0$ that is independent of $i$ and $j$ is a simplification
needed to make the model analytically tractable.\footnote{Note that
  heterogeneous volatilities can be resorbed into the definition of
  $s_{ij}\to \sigma_{ij}s_{ij}$. The discussion on firm size
  statistics still depends on the distribution of $s_{ij}$ alone, but
  growth rate statistics can be understood using the same theoretical
  machinery by considering the variable $\sigma_{ij}s_{ij}$ rather
  than $s_{ij}$ alone.} Using Equation~\eqref{eq:sub-growth}, we can
then write a firm's year-on-year growth rate as
\begin{equation}
  \label{eq:growth-def}
  g_{it} := \frac{S_{it+1} - S_{it}}{S_{it}} = \frac{1}{S_{it}}\sum_{j=1}^{K_i} s_{jit} \eta_{jit}\;\;\;.
\end{equation}

Within this framework, one could further assume that all the sub-unit
are of approximately the same size, say $\bar{s}_i$. In this case, a
firm's size and its number of sub-units would be directly
proportional, that is $S_i \approx K_i \bar{s}_i$. A direct application
of the Central Limit Theorem implies that for large firms,
and therefore firms with a large number of sub-units, the growth rates
$g_i$ are well described by a Gaussian distribution with a standard
deviation of $\sigma_0 S_i^{-\beta}$, with $\beta=1/2$. Thus, larger
firms would appear less volatile than smaller firms because they have
a higher number of sub-units, and are therefore better diversified.

This conclusion is, however, at odds with empirical observations,
which note that growth rates are markedly non-Gaussian, even when
large firms are studied. Further, the decay in volatility with respect
to firm size is much slower than in the anticipated
$\sigma_0 S^{-1/2}$, meaning that large firms are more volatile than
one would anticipate if all sub-unit sizes were approximately
equal. In the remaining of this section, we investigate the inadequacy
of this model in explaining the observed phenomena by considering two
versions of our statistical framework. The first one, due to
\cite{Gabaix_2011}, features a single source of granularity, namely
the size distribution of the sub-units composing a firm. The second,
due to \cite{wyart2003statistical}, adds another source of granularity
through the distribution of the number of such sub-units.

\subsection{Single granularity hypothesis [Gabaix' model]}
\label{sec:single_gran}

In this section, we depart from the assumption that sub-units are characterized by
relatively similar sizes $s_{ji}$ by positing that the
firms size distribution features a heavy tail following the
idea proposed in \cite{Gabaix_2011}.\footnote{Originally this
  granularity scenario was proposed to explain the anomalously large
  fluctuations of GDP in economies populated by firms whose size could
  be extremely large. Here we reinterpret a firm in the original paper
  as a sub-unit and a country as a firm. The remaining argument
  remains the same.} In line with the original model, we assume the
following.
\begin{assumption}
  \label{hp:su_size}
  The size of a sub-unit $s_{ji}$ is a random variable assumed to be
  larger than a lower bound $s_0$, representing the minimal size of a
  sub-unit. Its distribution is Pareto over the interval
  $[s_0;\infty)$:
  \begin{equation*}
    P(s_{ji}) = \frac{\mu}{s_0}\left(\frac{s_{ji}}{s_0}\right)^{-(1+\mu)}\;\; s_0>0,\;\; 1<\mu<2\;\;.
  \end{equation*}
The distribution $P(s_{ji})$ is taken to be time-invariant. Note that since $\mu > 1$, the average sub-unit size $\mathbb{E}[s_i]:=\bar{s}_i$ exists and is
finite.
\end{assumption}

To stay close to Gabaix's original model, where no explicit assumption is made on the statistics of sub-units number $K_i$, we assume that $K_i$ is close to $S_i/\bar{s}_i$ for large firms:
\begin{assumption}
  \label{hp:su_n_gabaix} For firms of size $S \gg \bar{s}$ and for $\mu > 1$, the number of sub-units $K$ is given by\footnote{More precisely
  $\big(\sum_{j=1}^{K} s_{ji} - K\bar{s}_i\big)/K^{1/ \mu}$ converges
  for large $K$ to a Lévy-stable random variable, see
  \citet{gnedenko1953}. This means that the typical fluctuations of
  $\sum_{j=1}^{K} s_{ji} - S$ are of order $K^{1/\mu}$. Since
  $S$ and $K$ are asymptotically proportional, this implies that
  the fluctuations of $S$ are of order $S^{1/\mu}$, and thus negligible except when $\mu \downarrow 1$. While mathematically convenient, this hypothesis looks at odds with the exiting evidence \citep{bose2006icc}. See Section~\ref{subsec:WB} for a version of the same model where $K$ is assumed to be itself a random variable. } 
  \[ 
  K = \frac{S}{\bar{s}} + o(S).
  \]
\end{assumption}
In our revisiting of Gabaix' model, we first discuss the relation between a
firm's size and the Herfindahl-Hirschman index (HHi) measuring its
concentration of sales among sub-units. We shall then study the distribution of the
growth rate volatility conditional on size, and how its shape
ultimately determines that of the growth rates when we pool together
heterogeneous firms.
\medskip

{\bf Variance of growth rates.} Now, since we assume that sub-units fluctuate independently, the variance of the growth rate of firm $i$, denoted as
${\sigma}^2_i$, is given by
\begin{equation}
  \label{eq:g_var}
  {\sigma}^2_i = \sigma_0^2 \frac{\sum_{j=1}^{K} s_{ji}^2}{\left(\sum_{j=1}^{K} s_{ji}\right)^2} := \sigma_0^2 \mathcal{H}\;\;\;, 
\end{equation}
where $\mathcal{H}$ represents the Herfindahl-Hirschman index (HHi) of
the sub-unit sizes of firm $i$. As noted by Gabaix, equation \eqref{eq:g_var} states that
the variance of the growth rates is proportional to an index measuring
the concentration of a firm's sales across its different sub-units. 
\medskip

{\bf Anomalous scaling behavior of $\mathcal{H}$}. The core result in
\cite{Gabaix_2011} can now be restated in our framework by saying that
when $1<\mu<2$ and $K \to \infty$, the typical value (i.e. the modal
value) of the Herfindahl-Hirschman index behaves as
$\mathcal{H}_{\text{typ}} \sim
\mathcal{O}(K^{2(1-\mu)/\mu})$.\footnote{See
  Proposition~\ref{lem:typical_herfindahl} in
  Appendix~\ref{app:proofs} which replicates Proposition 2 pp. 740 in
  \cite{Gabaix_2011} after noting that $\sigma=\sqrt{\mathcal{H}}$ and
  that $S$ is proportional to $K$. Throughout this paper, as in the
  orignal work, the symbol $\approx$ means approximately equal to, the
  symbol $\sim$ means that the ratio of between the left- and
  right-hand side terms tends to a positive real number as the
  argument tends to infinite. This can be considered as an asymptotic
  version of the $\propto$ symbol, representing proportionality
  between two terms.} An implication of this is that
$\mathcal{H}_{\text{typ}}$ decays to $0$ much slower than what one
would get in the case where $\mu>2$, where $\mathcal{H}_{\text{typ}}$
would behave as $\mathcal{O}(K^{-1})$. This last case leads to a
growth rate volatility that is proportional to the square root of the
HHi, $\sigma \propto \sqrt{\mathcal{H}}$, and therefore to the
scaling $S^{-1/2}$ described above.  Hence, in this granular setting
diversification is not as effective, since the existence of
particularly large and ``granular'' entities impedes the standard
diversification argument with a scaling in $S^{-1/2}$ and gives instead the
slow decay of the Herfindahl-Hirschman index described
above.\footnote{To complete the cases explored in the original paper,
  note that $\mathcal{H}_{\text{typ}} = \mathcal{O}(1/\ln K)$ when
  $\mu=1$. However in this case the assumption that $K = S/\bar{s}$ breaks down.}

The full study of the case where $1<\mu<2$ is in fact more subtle. The
problem, overlooked by ~\cite{Gabaix_2011}, lies in the fact that the
average value of $\mathcal{H}$ is much larger than the typical value,
$\mathbb{E}\left[\mathcal{H}\right] \gg \mathcal{H}_{\text{typ}}$. We
obtain this result by extending the approach proposed in
\cite{derridaRandomWalksSpin1997}, which we leverage to obtain the
asymptotic value of all integer moments
$\mathbb{E}[\mathcal{H}^k]$ with $k\in\mathbb{N}$, and to
characterize the probability distribution of the Herfindahl-Hirschman
index conditional on the number of sub-units, giving the following
proposition.

\begin{proposition}
  \label{prop:herfindahl_gabaix}
  For fixed $K\gg1$, the distribution of the Herfindahl-Hirschman
  index takes the following form
  \begin{equation}
    \label{eq:Herf}
    P(\mathcal{H}\vert K) = K^{2(\mu-1)/\mu} F\left(\mathcal{H}K^{2(\mu-1)/\mu}\right)\left(1-\sqrt{\mathcal{H}}\right)^{\mu-1}\;\;\;, 
  \end{equation}
  with $F(x) \sim \frac{1}{x^{1+\mu/2}}$ when $1 \ll x \lesssim K^{2(\mu-1)/\mu}$.
\end{proposition}
\begin{proof}
  See Appendix~\ref{ssec:proof_prop_gabaix1}.
\end{proof}
Proposition~\ref{prop:herfindahl_gabaix} defines the structure of the
conditional distribution of the HHi. The function $F(\cdot)$ is a scaling
function that is peaked approximately at $\mathcal{O}(1)$, ensuring
that the mode of the distribution $P(\mathcal{H}\vert K)$ is
$\mathcal{H}_{\text{typ}} \sim \mathcal{O}(K^{2(1-\mu)/\mu})$. Next,
using the distribution in Eq.\eqref{eq:Herf}, we can compute the
asymptotic behavior of all moments
$\mathbb{E}\left[\mathcal{H}^q | K \right]$ for any $q>0$, obtaining
the following proposition.
\begin{proposition}
  \label{prop:herfindahl_gabaix_moments}
  For all $q>0$, the moments of $\mathcal{H}$ conditional on $K$ are
  given, to the leading order for $K \gg 1$, by
  \begin{equation}
    \label{eq:gen_moments_body}
    \mathbb{E}\left[\mathcal{H}^q|K \right]\approx C_1 K^{1-\mu}+C_2K^{2q\frac{1-\mu}{\mu}} +\mathcal{O}\left(K^{\min\left(1-\mu, 2q \frac{1-\mu}{\mu}\right)}\right),
  \end{equation}
  where $C_1$ and $C_2$ are two constants. 
\end{proposition}
\begin{proof}
See Appendix~\ref{ssec:proof_prop_gabaix2}.
\end{proof}
This proposition states that the conditional moments of the HHi
contain two leading contributions: the first behaving as $K^{1-\mu}$
and the second one as $K^{2q\frac{1-\mu}{\mu}}$, that is as
$\mathcal{H}^{q}_{\text{typ}}$. The former is dominant whenever
$1-\mu>2q\frac{1-\mu}{\mu}$, i.e. $q>\mu/2$, while the latter is dominant
when $q < \mu/2$. 

Since we work under the assumption that $1<\mu<2$,
the first term in Eq.~\eqref{eq:gen_moments_body} is dominant for
$q=1$, showing that $\mathbb{E}\left[\mathcal{H}|K \right]$ is larger
than its typical value $\mathcal{H}_{\text{typ}}$. When $q=1/2$, on
the contrary, it is the second term that dominates and
$\mathbb{E}[\sqrt{\mathcal{H}}|K]$ is driven by the typical value of
the HHi. Hence, in summary, we find that 
\begin{equation}
  \label{eq:sigma_scal_H}
  \sqrt{\mathbb{E}[\mathcal{H} \vert K]} \sim K^{(1-\mu)/2} \qquad \qquad
  \mathbb{E}[\sqrt{\mathcal{H}} \vert K] \sim K^{(1-\mu)/\mu}\;\;.
\end{equation}
In plain English, this means that the scaling relation one observes is
different if one considers the square-root of the average HHi or the
average of the square-root of the HHi. The economic interpretation
behind this unexpected result is that the model describes a market
where two types of large firms coexist. The first type is composed by
poorly diversified firms with an abnormally high HHi
($\mathcal{H}\approx 1$), namely firms with one or two dominant
sub-units in terms of sales. These high values of the HHi drive the
behavior of $\sqrt{E[\mathcal{H}|K]}$. The second type is composed, on
the contrary, by well diversified large firms representing the typical
case. These are dominant in driving
$\mathbb{E}[\sqrt{\mathcal{H}} \vert K]$. The existence of the first
type of large firms, overlooked in \cite{Gabaix_2011}, is at the
origin of the anomalous scaling exponents. Further it has important
consequences on the behavior of the growth volatilities $\sigma_i$ and
ultimately on the shape of the growth rates distribution. Two aspects
we investigate below.\medskip

{\bf Conditional distribution of the growth rate volatility}. Since
the volatility of the growth rate is equal to
$\sigma_0 \sqrt{\mathcal{H}}$, and since $K$ and $S$ are proportional
($S \approx K \overline{s}$) all the results concerning the
Herfindahl-Hirschman index naturally extend to the growth volatility.
The anomalous scaling in Eq.~\eqref{eq:sigma_scal_H} now becomes
\begin{equation}
  \label{eq:sigma_scal}
  \mathbb{E}[\sigma \vert S] \approx \sigma_0 \left(\frac{S}{\bar{s}}\right)^{(1-\mu)/\mu} \qquad
  \sqrt{\mathbb{E}[\sigma^2 \vert S]}  = \sigma_0 \left(\frac{S}{\bar{s}}\right)^{(1-\mu)/2} \;\;\;,
\end{equation}
showing again two different values of the exponent in the scaling
relation depending on whether one considers the average volatility of
growth rates or the square-root of the average squared volatility.

Similarly, the conditional distribution of volatilities can be deduced
from Eq. \eqref{eq:Herf} by substituting
$\sigma = \sigma_0 \sqrt{\mathcal{H}}$ to obtain
\begin{equation}
  \label{eq:Herf2}
  P(\sigma \vert S) =
  \frac{1}{\bar{\sigma}(S)}
  G\left(\frac{\sigma}{\bar{\sigma}(S)}\right) \left(1 -
    \frac{\sigma}{\sigma_0}\right)^{\mu - 1} , \quad\text{with} \quad
  G(x)=2x F(x^2),
\end{equation}
where $F(\cdot)$ is the same function defined in Eq. \ref{eq:Herf} above and $\bar{\sigma}(S):=\mathbb{E}[\sigma \vert S]$

Note that Eq.~\eqref{eq:Herf2} implies that $P(\sigma|S)$ is a
truncated power-law: indeed,
$P(\sigma|S)\approx G\left(\frac{\sigma}{\bar{\sigma}(S)}\right)$ when $\sigma \ll \sigma_0$, and therefore
$P(\sigma|S)\propto \sigma^{-1-\mu}$ in the regime
$\bar{\sigma}(S)\ll \sigma \ll \sigma_0$. This is of course
valid up to $\sigma=\sigma_0$, indicating the extreme case where the
sales of a firm are concentrated in a single subunit, leading to
$P(\sigma_0|S)=0$ and establishing the cut-off of the power-law. This
reflects the two types of firms: those close to the mode which have
$\sigma\approx \bar{\sigma}(S)$ and those which are in the
tail and are therefore poorly diversified.  The entire distribution of
the volatilities $P(\sigma)=\int \dint S~P(S)P(\sigma|S)$ can be
recovered from this, and will retain the truncated power-law nature of
the conditional distributions. A qualitative representation of
$P(\sigma|S)$ is reported in Figure~\ref{fig:sigma_qualitative}
(left-panel), highlighting the two types of firms that shape the
distribution.

 \medskip


\begin{figure}[t]
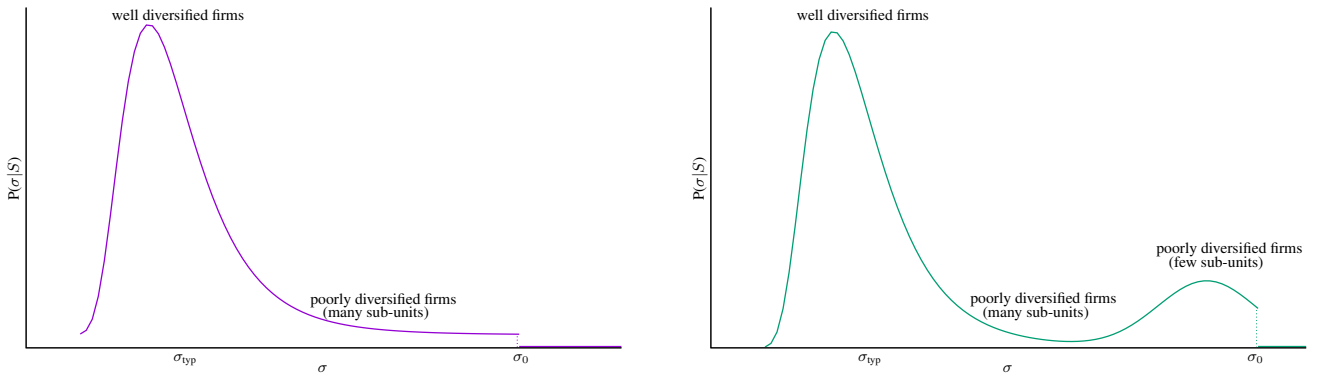

  \caption{Qualitative behavior of $P(\sigma|S)$}
  \label{fig:sigma_qualitative}
  \begin{center}    
    \begin{minipage}[t]{0.47\linewidth}  
    \scalebox{0.65}{
      \input{./figures/fig_PsigmaS_gabaix.tex}
      }
    \end{minipage}
    \hfill
    \begin{minipage}[t]{0.47\linewidth}  
    \scalebox{0.65}{
      \input{./figures/fig_PsigmaS_wb.tex}
      }
    \end{minipage}
  \end{center}
  {\scriptsize {\it Notes}: \underline{Left-panel} reports the
    qualitative behavior of the distribution of the growth rate
    volatility $P(\sigma|S)$ with a single granularity source defined
    in Eq.~\ref{eq:Herf2}. The bulk of the distribution is constituted
    by firms that are well diversified with an HHi around
    $\sigma_{\text{typ}}$. The power law right tail up to
    $\sigma \approx \sigma_0$ is generated by poorly diversified firms
    with a very unequal internal size distribution across
    sub-units. At $\sigma_0$ the power law tail is
    truncated. \underline{Right-panel} reports the qualitative
    behavior of the distribution of the growth rate volatility
    $P(\sigma|S)$ with two granularity sources in
    Eq.~\eqref{eq:cond_vol}. In this case there is an additional
    contribution, a peak close to $\mathcal{H}\approx 1$ of relative
    total weight of order $S^{\alpha-\mu}$, corresponding in that case
    to large firms that are made up of a very small number of
    sub-units. In both panel the qualitative
    behavior has been reproduced using an Inverse Gamma
    distribution. \par}
\end{figure}

\subsection{Double granularity hypothesis [Wyart and Bouchaud's model]}
\label{subsec:WB}
In this section we extend our statistical framework adding a second
source of granularity associated with the number of sub-units
following the idea introduced in \cite{wyart2003statistical} (WB
hereafter). While Assumptions~\ref{hp:gibrat} and \ref{hp:su_size}
remain the same, we modify Assumption~\ref{hp:su_n_gabaix} by stating
that the number of sub-units is itself a random variable distributed
as a Pareto. Formally, this is stated as follows.
\begin{assumption}
  \label{hp:3}
  The number of sub-units $K_i$ is random variable assumed to be distributed according to Pareto over the interval
  $[1;\infty)$,
  \begin{equation*}
    P(K_{i}) = {\alpha} K_{i}^{-(1+\alpha)}\,;\quad 1<\alpha<\mu<2.
  \end{equation*}
  The distribution $P(K)$ is taken to be time-invariant.
\end{assumption}
Note first that, while the number of sub-units $K_i$ is an integer, we
model it as a continuous random variable for mathematical
convenience. The results we provide in this section do not depend on
the minute details of the distributions of $K_i$, but only on the
behavior of its tail.
In the following, we take Assumption \ref{hp:3} to be true for all of
our derivations, which lets us clarify and generalize the results
found in \cite{wyart2003statistical}.

\medskip

{\bf Firm size distribution}. We start our analysis by focusing on the
firm size distribution, an aspect that could not be investigated
within the single granularity set-up because of the lack of a specific
assumption on $P(K_i)$. The behavior of the firm size distribution in
a model with granularity both in the size and number of sub-units is
presented in the following proposition.
\begin{proposition}
  \label{prop:fsd}
  The firm size distribution behaves, for large $S$, as the sum of the
  two Pareto distributions
  \begin{equation} 
    P(S_i) \sim \frac{C_\alpha}{S_i^{1+\alpha}} + \frac{C_\mu}{S_i^{1+\mu}},
  \end{equation}
  where $C_\alpha = \left(\frac{\mu}{\mu-1}\right)^\alpha$ and
  $C_\mu=\frac{\alpha}{\alpha-1}$.
\end{proposition}
\begin{proof}
  See Appendix~\ref{app:fsd}.
\end{proof}
Since $\alpha < \mu$, this proposition implies that the first term
dominates and determines the tail of the size distribution. This case
corresponds to a scenario where large firms are most likely made up of
a large number of sub-units $K \sim S/\overline{s}$ and the tail of
$P(S)$ simply mirrors that of the distribution $P(K)$. The second
term, on the other hand, corresponds to much rarer cases where large
firms are made up of a handful of large sub-units (hence explaining
why the tail behavior has an exponent $\mu$), and are thus poorly
diversified. The relative probability that large firms are in this
second category, we denote with $\pi(S_i)$, behaves as
$(C_\mu/C_\alpha) \, S_i^{\alpha - \mu}$ and goes to $0$ for large
$S_i$. Although it does not provide an explanation to why the number
and size of the sub-units should be Pareto distributed,\footnote{There
  exists a large class of general and intuitive processes that
  generate Pareto distributions in both $K_i$ and $s_{ji}$. One
  example is to assume that the sub-units of a firm grow according to
  a mixture of stochastic growth and redistribution between them as in
  $s_{ijt+1}- s_{ijt}=\eta_{ijt+1}\,s_{ijt}+\gamma\left(\frac{1}{K_i}
    \sum_{k=1}^{K_i} s_{ikt} - s_{ijt} \right)$ where the first term
  accounts for stochastic multiplicative growth and the term with
  $\gamma$ for a net flow going from larger sub-units into smaller
  sub-units. See \cite{Bouchaud2000} for details.} Proposition
\ref{prop:fsd} clarifies that the tail of the firm size distribution
inherits its behavior from the distribution featuring the heavier tail
between those of the sub-unit size and number. \medskip

{\bf Conditional distribution of the growth rate volatility}. We next
move to the distribution of growth rate volatility conditional on
size.
\begin{proposition}
  \label{prop:vol_dist}
  The distribution of growth rate volatilities conditional on size $S$
  is given, for large $S$, by:
  \begin{equation}
    \label{eq:cond_vol}
    P(\sigma \vert S) \approx
    \frac{1-\pi(S)}{\bar{\sigma}(S)}
    G\left(\frac{\sigma}{\bar{\sigma}(S)}\right)\left(1-\frac{\sigma}{\sigma_0}\right)^{\mu-1}
    +\pi(S)H(\sigma),
  \end{equation}
  where $\bar{\sigma}(S) \sim S^{-\beta}$ is given by
  Eq. \eqref{eq:sigma_scal} with $\beta=(\mu-1)/\mu$,
  $\pi(S)\sim S^{\alpha-\mu}$ and $H(\cdot)$ is a contribution peaked
  at $\sigma \approx \sigma_0$. Finally $G(\cdot)$ is defined in
  Eq. \eqref{eq:Herf2} with $G(x)\sim x^{-1-\mu}$.
\end{proposition}
\begin{proof}
  See Appendix~\ref{app:vol_dist}.
\end{proof}
This proposition states that conditional distribution of the
growth rate volatility is a mix of two components weighted with
probability $1-\pi(S_i)$ and $\pi(S_i)$.

The first component reproduces the volatility distribution observed in
the previous section under the single granularity hypothesis with the
two regimes associated with well and poorly diversified large firms
featuring an evenly or skewed distribution of sales among
sub-units. This means that the conditional volatility distribution
displays a power-law tail regime $\sigma^{-1-\mu}$ with $1 < \mu < 2$
for $\bar{\sigma}(S)\ll \sigma \ll \sigma_0$. Moreover, the
growth volatility scales with size as $S^{-\beta}$ with
$\beta=(\mu-1)/\mu$ and not as $S^{-1/2}$ as predicted under the
standard Central Limit theorem. Note that all these firms typically
have a large number of sub-units.

The second component is new and represents a peculiar feature that
emerges when the two granularity sources are combined, something that
was overlooked by both \cite{wyart2003statistical}
and~\cite{Gabaix_2011}. With probability $\pi(S_i)$, large firms can
be constituted by a very low number of sub-units, which mechanically
concentrate aggregate sales and prevent diversification. For these
firms, the volatility remains of order $\sigma_0$, generating an hump
in the distribution around this value. The qualitative behavior of
$P(\sigma|S)$ is reported in the right-panel of
Figure~\ref{fig:sigma_qualitative}.

Proposition~\ref{prop:vol_dist} gives the opportunity to clarify the
subtle relation between the two set-ups with either one or two sources
of granularity. The former can be imagined as fixing the number of
sub-units $K$ to $K = S/\overline{s}$, with the consequence of
removing the ``hump'' in the volatility distribution around
$\sigma_0$. Indeed, the probability to observe a ill-diversified firm
in a single granularity model is of order $S^{-1 -\mu}$, which tends
to zero much faster than $S^{\alpha - \mu}$, obtained in the case with
two sources of granularity.

Using the distribution derived in Proposition~\ref{prop:vol_dist} we
can proceed to study the behavior of its moments.
\begin{proposition}
  \label{prop:vol_moments}
  For $1 \leq \alpha < \mu$, the integer moments of the growth rate
  volatilities conditional to size $S$ are asymptotically given, for
  large $S$, by:
  \begin{equation}
    \label{eq:vol_moments}
    \mathbb{E}[\sigma^q|S] = C_1 S^{1-\mu} + C_2 S^{q\frac{1-\mu}{\mu}} + C_3 S^{\alpha-\mu} +  \mathcal{O}\left(S^{\min \left(\alpha-\mu, 1-\mu, q\frac{1-\mu}{\mu}\right)}\right)\;\;,
  \end{equation}
  where $C_1$, $C_1$ and $C_3$ are numerical constants.
\end{proposition}
\begin{proof}
  See Appendix~\ref{app:vol_moments}.
\end{proof}
This proposition states that the dominant behavior of the conditional
moments of the volatility is determined by the smallest exponent in
the expansion above. When $1 < \alpha < \mu$, all the moments of order
greater or equal to $2$ ($q\geq 2$) are asymptotically dominated by
the third term, $C_3 S^{\alpha-\mu}$, which represents the
contribution of large firms with few sub-units. This also happens with
the first moment ($q=1$) if the tail of the distribution of the number
of sub-units $P(K)$ is light and satisfies
$\alpha > \mu +\frac{1-\mu}{\mu}$. The naive scaling for the first
moment, obtained when the second term in Eq.~\eqref{eq:vol_moments} is
dominant, only holds when the right tail of $P(K)$ is heavy enough,
with $\alpha < \mu +\frac{1-\mu}{\mu}$.
Note finally that the {\it median} volatility always scales as
$\sigma_{\text{med}}(S) \sim S^{-\beta}$ with
$\beta=(\mu-1)/\mu$. This proposition again confirms that the scaling
behavior of volatility moments in this model depends on the moment
under consideration.

The economics of this proposition can be understood by noticing that
firms with large size $S$ can be in this set-up of three different
types. The first type is composed by large firms whose sales are
evenly distributed among a large number of sub-units. For these firms
$S \approx K \bar{s}$ and the HHi is relatively small
$\mathcal{H}\approx \mathcal{H}_{\text{typ}}$, meaning they contribute
to the moments of the growth volatility by a factor of order
$S^{q\frac{1-\mu}{\mu}}$. The second type is composed by large firms
that also have a large number of sub-units. However, the sales of
these firms are concentrated in very few sub-units. Again for these
firms $S \approx K \bar{s}$, but their HHi is very high and can
go up to $\mathcal{H}\approx 1$; their contribution to volatility
moments behaves as $S^{1-\mu}$ independently of $q$. Finally, a
fraction of these large firms is made up of only a few sub-units which
mechanically results in a very high Herfindahl-Hirschman index
($\mathcal{H}\approx 1$). They therefore will have
$\sigma \approx \sigma_0$ and all moments will scale as $\sigma_0^q$
independently of $S$. They contribute to volatility moments by a
factor proportional to their relative fraction, $S^{\alpha-\mu}$. This
last group was not present in \cite{Gabaix_2011} model and was
overlooked in \cite{wyart2003statistical}. \medskip

{\bf Distribution of the growth rates}. We are now able to provide a
complete characterization of the probability distribution of the
growth rates. We start with the following proposition.
\begin{proposition}
  \label{prop:gmm_levy}
  The distribution of growth rates conditional on size $S$ and on
  growth volatility $\sigma$ is given, for large $S$, by:
  \begin{equation}
    \label{eq:WB_GMM}
    P_S(g \vert \sigma,S) \approx \left(1 - \pi(S)\right)\;\mathcal{N}(0,\sigma^2) +  \pi(S)\;Q_{\eta} \;\;,
  \end{equation}
  where $\mathcal{N}(0,\sigma^2)$ is a Normal distribution with
  variance $\sigma^2=\sigma^2_0\mathcal{H}$ and $Q_{\eta}$ a non
  universal distribution that depends on the distribution of the
  sub-unit growth shocks $\eta$. The weight $\pi(S)$ represents the
  probability of observing a large firm with only few sub-units
  vanishing when size grows larger as $S^{\alpha-\mu}$.

  Neglecting large firms with a small number of sub-units and
  integrating Eq.~\eqref{eq:WB_GMM} over the first term of the
  distribution $P(\sigma|S)$ in Eq.~\eqref{eq:cond_vol} gives
  \begin{equation}
    \label{eq:pred2_gr}
    P\left(g|S\right) \sim
    S^{\beta}L_{\mu, 0}(gS^{\beta}),\;\;\; \text{when} \;\; g \ll 1,
  \end{equation}
  where $L_{\mu,0}(\cdot)$ is the symmetric L\'evy alpha-stable
  probability density with stability parameter
  $1 < \mu < 2$.\footnote{Symmetric L\'evy alpha-stable distributions
    are a family of distributions defined with location, scale and
    stability parameters. This last parameter determines the behavior
    of the tails of the distribution. When the stability parameter is
    equal to 2, one obtains a Gaussian distribution, which is a subset
    of this family. Instead, taking a stability parameter of 1 leads
    to the Cauchy distribution.} Because of the cut-off in the
  distribution of $\sigma$, this distribution also has a cut-off, with
  $P(g|S)=0$ for $g\gtrapprox S$. The complete distribution $P(g)$ is
  obtained by integrating over $P(s)$, and behaves asymptotically as
  $P(g)\sim \vert g\vert^{-1-\mu}$.
\end{proposition}
\begin{proof}
See Appendix~\ref{app:gmm_levy}.
\end{proof}

\medskip

Before discussing the content of this proposition, recall that a
firm's growth rate is defined as the sum of the sub-unit growth shocks
weighted by their relative size
(cfr. Eq.~\eqref{eq:growth-def}). Growth volatility is therefore
proportional to the square root of the HHi, computed using the sizes
of the sub-units (cfr. Eq.~\eqref{eq:g_var}). With this in mind,
Proposition~\ref{prop:gmm_levy} states that the growth rate
distribution, conditional on size and volatility, is again the
combination of two components.

The first component is associated with of large firms made up by a
large number of sub-units, and results from invoking the Central Limit Theorem (CLT henceforth). In the
limit of a large number of sub-units, a firm's growth rate $g$ is the
sum of a large number of random variables with finite variance; its
distribution converges therefore to a Normal distribution. The speed
of convergence to this asymptotic limit is not uniform, and depends on
the extent to which a firm's size is evenly distributed among its
sub-units. However, for a large enough number of sub-units, the Normal
distribution provides a good approximation independently of the
distribution of the $\eta$ shocks, as long as it has a finite
variance. The second component is instead associated with large firms
having only a handful of sub-units, thus preventing the CLT from
applying. Their growth rate distribution is therefore strongly
determined by that of the sub-unit shocks, $\eta$. In the extreme case
of a firm with one single sub-unit, the two distributions would
coincide.

The second statement in Proposition~\ref{prop:gmm_levy} concerns the
behavior of the growth rate distribution when we disregard the role of
large firms with few sub-units, and we pool together those with a
large number of sub-units but with different growth volatilities
(stemming from their very heterogeneous level of
diversification). Formally, this exercise consists in mixing together
random variables, which are approximately Gaussian, using the
conditional volatility distribution $P(\sigma|S)$ as a mixing
function.\footnote{This only concerns the first term of $P(\sigma|S)$
  in Eq.~\eqref{eq:cond_vol} since we are disregarding large firms
  with few sub-units.} The resulting distribution,
$P(g|S)\approx \int \mathrm{d}\sigma~ \mathcal{N}(0,\sigma) P(\sigma|S)$,
inherits its behavior from $P(\sigma|S)$ which has a Pareto right tail
with exponent $-(1+\mu)$, and is truncated at
$\sigma_0$. Correspondingly, our proposition states that $P(g|S)$ is
distributed according to a {\it truncated} symmetric L\'evy
alpha-stable distribution with a stability parameter $\mu$. This is
consistent with a power law decay with exponent $-(1+\mu)$, truncated
for large growth rates.

Taking stock, we can conclude by discussing the behavior of the
unconditional growth rates distribution $P(g)$, which is obtained when
we pool together all types of firms. Their heterogeneity in both
structure and size is captured by their heterogeneity in growth
volatility. This distribution is defined in terms of a Gaussian
mixture as
$P(g) \approx \int {\rm d}S {\rm d}\sigma \, \mathcal{N}(0,\sigma) \,
P(\sigma \vert S) P(S)$. Hence, the behavior of the growth rates
distribution in this model can be rationalized in terms of a class of
statistical models known as Gaussian Mixtures. The basic idea is that
mixing together random variables with centered distributions but
heterogeneous variances provides a simple and intuitive mechanism to
generate continuous, unimodal and heavy-tailed distributions such as
the one observed in the
data. ~\citep{andrewsmallows1974,west1987}. Such mixtures have also
been proposed to explain firm growth rate distributions
in~\cite{Buldyrev_et_al_2007}.

If this statement holds, then un-mixing the growth rates by rescaling
them by the corresponding individual volatility, would lead to an
approximately Normal distribution. The quality of this approximation
is driven by two factors. First, it's possible that we would observe
large firms with few sub-units, whose non-Normal growth rates
contaminate the Gaussian mixture. Second, the very heterogeneous sizes
of the sub-units which make up a large firm slows down the convergence
to the Normal limit induced by the CLT.

\medskip

{\bf Aggregation}. Before concluding this
theoretical section we present a last result describing how the model
behaves against aggregation.
\begin{proposition}
  \label{pred:agg}
    Proposition~\ref{prop:fsd}, \ref{prop:vol_dist},
  \ref{prop:vol_moments} and \ref{prop:gmm_levy} are robust against
  the additive aggregation of firms into (possibly fictitious)
  supra-firms, or even the whole economic activity of a country.
\end{proposition}
\begin{proof}
  See Appendix~\ref{app:agg}.
\end{proof}
This result has the notable implication the statistical properties we
have described above are independent of the level of aggregation of
the original sub-units. Merging firms into supra-firms would generate
entities whose growth trajectories would be described in a first
approximation by Proposition~\ref{prop:fsd},
 Proposition~\ref{prop:vol_dist}, Proposition~\ref{prop:vol_moments}
and Proposition~\ref{prop:gmm_levy}.\footnote{This is an important
  property of the model since is in line with the available empirical
  evidence \citep{Lee_Amaral_Canning_Meyer_Stanley_1998,Castaldi2008}
  and with Gabaix' granularity hypothesis.  Moreover while this
  robustness against aggregation is shared with Gabaix' model, it
  possesses some discriminatory power: for example, the
  \cite{Sutton_2002} model does not satisfy this property
  \citep{wyart2003statistical}.} 

\section{Empirical investigations}
\label{sec:empirics}

The models analyzed in this paper attempt to explain key statistical
regularities in the growth of business firms. Both surmise that these
regularities are a consequence of the internal structure of
firms. Heterogeneous and granular distributions for the number and
sizes of the sub-units of a firm have a direct impact on its growth
volatility. In these models, firms can be separated into groups, with
well-diversified firms whose size is evenly distributed among several
sub-units on the one hand, and firms which are poorly diversified on
the other. The second group's poor diversification results either from
them being made up of a small number of sub-units, or from a large
number of sub-units but whose size is concentrated in only a handful
of them. Ultimately, the coexistence of these different firms in a
market should determine how their size and growth interact.

These new results translate into more precise testable
predictions. Specifically, we focus on three aspects that have
received little attention in the literature: the distribution of the
growth rate volatilities, its relation with firm size, and finally how
it ultimately determines the shape of the growth rate distribution. In
performing our investigations one should always keep in mind that our
predictions are only sharp in the limit of large firm sizes,
$S \to \infty$. Therefore empirical analysis are expected to show some
deviations from these predictions, which should nevertheless taper off
as $S$ increases. Before diving in these investigations, we begin by
describing our data, empirical proxies and statistical procedures.

\subsection{Data, empirical proxies and binning procedure}
\label{sec:data}

Our empirical investigations take advantage of the Compustat database
collecting financial, statistical and market information on global
companies throughout the world. Compustat collects data for medium to
large publicly traded companies, and is well known for the quality of
its data and for the long time-span it covers. While we are aware of
the limitation of the Compustat data in terms of sample selection, we
believe that they are high quality data on the US economy with legal
entities which are consistently defined in time and breadth, along
with high frequency observations which are not typically available in
business register data.

From this source, we extract a panel of companies spanning the period
from 1961 to 2020 with information collected at the quarterly
frequency. For each company, we observe size as proxied by Total
Sales,\footnote{The Compustat name for the size proxy is 'saleq'.} and
we deflate nominal sales using the quarterly GDP deflator provided by
the Federal Reserve Bank of St. Louis.\footnote{We download the
  'GDPDEF' quarterly series from the Federal Reserve Economic Data
  (FRED) web site.} To limit inconsistencies regarding our size proxy
due to different accounting rules, we consider only companies
incorporated in the US.

To clarify our data pre-processing, let $\tilde{S}_{iqy}$ represent the
size of a firm $i$ in a quarter $q$ and year $y$ expressed in millions
of US dollars deflated with a GDP deflator with base year in 2012. We
define a normalized quarter size as
$S_{iqy}=\frac{N_y\tilde{S}_{iqy}}{\sum_{iq} \tilde{S}_{iqy}}$ where
$N_y$ represent the number of quarter size observations available in
year $y$. This normalization is important stationarize the size
distribution. To simplify the notation, we absorb $qy$ into a time
subscript $t$. Since our data records firm size at the quarterly
frequency we define annual growth rates as the logarithmic difference
over one year, $g_{it}=\ln(S_{it+4}) - \ln(S_{it})$, rolling over
quarters to maximize the number of observations. Similarly when we
work with the growth rate of the non-normalized size
$\tilde{S}_{it}$. For each firm we then compute a single averaged
(across years) normalized size $S_i$, and we investigate its growth
dynamics considering the corresponding annual growth rates
$g_{it}$. After removing firms for which we do not observe at least
two growth rates, we are left with $24,233$ (out of a total of
$38,995$) companies, corresponding to over one million growth rate
observations. Companies in the Compustat sample tend to be large,
since they are publicly traded, with an average quarterly Total Sales
value of about 450 millions in 2012 USD and a positive average growth
rate. For each firm we observe on average 46 growth rates, a number
that should limit concerns about the potential noise affecting the
firm growth volatility estimates we compute in the paper. Descriptive
statistics on the sample used are reported in Appendix~\ref{app:main}.

Finally, to reduce the effect of individual observations when we
investigate the relation between a firm's size and its growth
dynamics, we implement a pre-processing technique known as data
binning. This technique consists in grouping firms based on their size
and compute for each group (i.e. bin) some statistics of the
corresponding growth rates. In the following, we set the number of
bins to $25$ and we assign each firm to each bin based on its average
normalized size $\bar{S}_{i}=\frac{1}{T_{i}} \sum_t S_{it}$. The bins
are defined so that they contain the same number of firms. We then
compute for each firm the corresponding growth volatility, proxied by
the adjusted mean absolute deviation
$\sigma_i =\sqrt{\frac{\pi}{2}}\frac{1}{T_{i}} \sum_t
|g_{it}-\bar{g}_{i}|$, where $\bar{g}_{i}$ is the sample average
growth rate of firm $i$ computed over $T_i$ observations.\footnote{A
  bar over a variable indicates from now on the corresponding sample
  average. In this definition, the correcting factor
  $\sqrt{\frac{\pi}{2}}$ comes from the observation that for a
  centered Normal distribution
  $E[\vert X\vert]/\sqrt{E[X^2]}=\sqrt{\frac{2}{\pi}}$, so that the
  mean absolute deviation is a low-moment estimator for the standard
  deviation.} We then compute the average size and volatility across
firms belonging to the same bin. We implement this procedure also for
higher absolute moments of the volatility we denote as $\sigma_i^q$
with $q=2,3,4$.

\subsection{Growth rate volatility}
\label{sec:gr-volatility}

Our first investigation concerns the statistical properties of the
growth rate volatility $\sigma_i$. Indeed,
Proposition~\ref{prop:vol_dist} describes the distribution of these
volatilities conditional on size, allowing to make the following
prediction.
\begin{prediction}
  \label{pred:1}
  For large size $S$, the distribution of growth rate volatilities
  conditioned on size, $P(\sigma|S)$, displays a size-dependent
  contribution together with a size-independent contribution that
  appears as a hump at $\sigma=\sigma_0$. Rescaling the volatilities
  by their average computed over firms of the same size,
  $\bar{\sigma}(S)$, leads to a curve collapse for the size-dependent
  contribution into a single master curve displaying a power law tail
  with exponent $-(1+\mu)$.
\end{prediction}
To test this prediction, we start by computing the growth volatility
$\sigma_i$ for each of the $24,233$ firms in the sample according to
the procedure described in Section~\ref{sec:data} above. Volatilities
range from $0.002$ to $5.879$ with an average value of $0.479$ and a
standard deviation of $0.490$. We then group volatilities in 25 size
bins containing approximately $970$ observations each and look at
their distributions within the bins. Left-panel of
Figure~\ref{fig:pred1} displays the estimated probability density
(histogram) for the 25 bins on a double-log scale. Moving from smaller
(dark-violet) to larger (gold) firms the distribution shifts smoothly
towards the left, pushing the corresponding average volatility from
$1.218$ to $0.219$. The right tail in all the 25 distributions appears
to be log-linear, suggesting a power law decay of their densities. We
note, however, the absence of a hump for large volatilities
representing the size-independent contribution in
Prediction~\ref{pred:1}. Figure~\ref{fig-app:pred1_hump} in
Appendix~\ref{app:main} reports the distribution of the growth rate
volatility on a double log scale for the six size bins 1, 5, 10, 15,
20 and 25. We see no sign of the hump for large volatilities. We
confirm our visual inspection with a non-parametric test of
bi-modality, which rejects the existence of a second mode at any
reasonable level of statistical significance.
\begin{figure}[t]
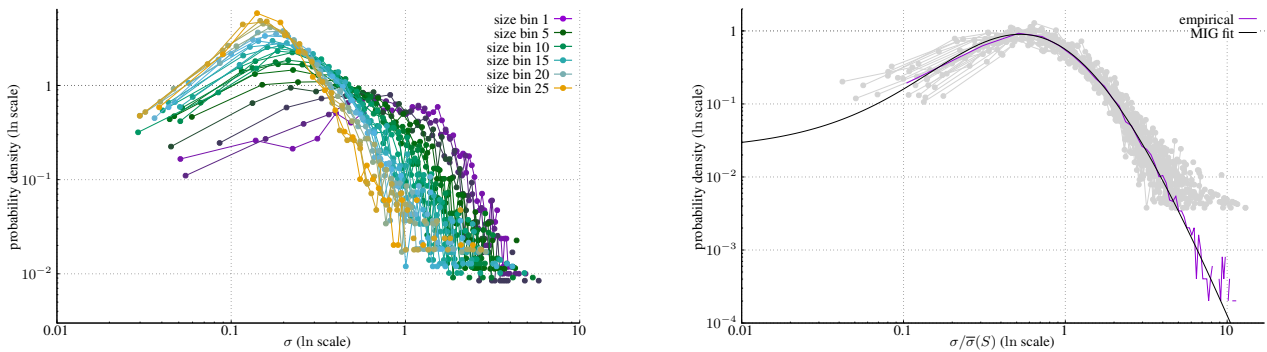

  \caption{Growth volatility distribution}
  \label{fig:pred1}
  \begin{center}    
    \begin{minipage}[t]{0.47\linewidth} 
    \scalebox{0.6}{
      \input{./figures/fig_MAD_binned.tex}
      }
    \end{minipage}
    \hfill
    \begin{minipage}[t]{0.47\linewidth}
    \scalebox{0.6}{
      \input{./figures/fig_MAD_rescaled_binned.tex}
      }
    \end{minipage}
  \end{center}
  {\scriptsize {\it Notes}. \underline{Left panel} reports on a double
    log scale and for 25 bins, defined in term of normalized size, the
    distribution of the growth rate volatility. \underline{Right
      panel} reports on a double log scale the distribution of the
    growth volatility rescaled by the average volatility observed in
    each bin together with a Lognormal (black line) and Inverse Gamma
    fit (solid line). The ML estimates of the scale, shape and
    location parameters in the Inverse Gamma fit are found
    $4.788(0.139)$, $4.620(0.086)$ and $0.326(0.010)$ respectively. In
    all panels volatility is computed as the mean absolute deviation
    multiplied by the factor $\sqrt{\pi/2}$. Data source:
    Compustat. \par}
\end{figure}

We next construct the rescaled volatities, $\sigma_i/\bar{\sigma}(S)$,
and report the corresponding empirical densities in the right-panel of
Figure~\ref{fig:pred1}. In line with Prediction~\ref{pred:1}, all the
25 distributions of the rescaled volatility (light-gray shaded curves)
collapse nicely onto a single master curve. We find again no sign of
the hump at large volatilities, indicating that $P(\sigma|S)$ is
entirely determined by the size dependent contribution in
Eq.~\eqref{eq:cond_vol}. However, in this case one should consider
that the location of the hump in the 25 distributions would move
around with the rescaling by $\bar{\sigma}(S)$ because the relative
size of the contribution would shift as well, making harder the
identification of the humps.

Taking stock of this behavior, we investigate the properties of the
distribution of the rescaled growth volatility, corresponding to the
function $G$ in Eq.~\eqref{eq:cond_vol}. We achieve this by pooling
together the rescaled volatilities and estimating the resulting
density. The histogram we obtain, shown in Figure~\ref{fig:pred1}
(dark-violet line), matches the shaded gray area representing the master
curve on top of which the 25 binned rescaled histograms collapse. We
propose to characterize this empirical histogram with a 3-parameters
Modified Inverse Gamma (MIG) density,
\begin{equation}
  \label{eq:invgamma}
  P_{MIG}(x;a,b,m) = C(a,b,m) (x+m)^{-(1+b)}\exp\left(- \frac{a}{x+m}\right), \qquad (x \geq 0),
\end{equation}
where $(a,b,m)$ represent the scale, shape and location parameter
respectively. The normalization constant $C(a,b,m)$ normalizes the
distribution when integrating over $\sigma\geq 0$.\footnote{The
  normalization constant in the Modified Inverse Gamma density reads
  $\dfrac{a^b}{\Gamma(b)}$, while in Eq.~\eqref{eq:invgamma} we set
  $C(a,b,m)=\dfrac{a^b}{\Gamma(b)-\Gamma(b,\frac{a}{m})}$ where
  $\Gamma(b,\frac{a}{m})$ is the upper Incomplete gamma function.} The
reason for this choice is that the Modified Inverse Gamma
distribution, as the original Inverse Gamma, features a power law
right-tail with exponent equal to $-(1+b)$ for large volatilities. The
specific choice of a Modified Inverse Gamma distribution is otherwise
purely phenomenological and we do not attribute too much meaning to
the small $\sigma$ behaviour, for which the collapse of the different
curves shown in Figure~\ref{fig:pred1} is far from perfect. It is
important however that the right tail of the Inverse Gamma is
consistent with the power law tail of function $G$ in
Eq.~\eqref{eq:cond_vol} resulting, in turn, from the tail of the
function $F$ in Eq.~\eqref{eq:Herf}.

A Maximum Likelihood estimation of the parameters yields
$4.788(0.139)$, $4.620(0.086)$ and $0.326(0.010)$ for the scale, shape
and location parameters respectively. The fit, reported in the
right-panel of Figure~\ref{fig:pred1} (solid black-line), describes
its empirical target rather well. We conclude that the Modified
Inverse Gamma distribution is a good first approximation to describe
the function $G(x)$ in Eq.~\eqref{eq:cond_vol}. Moreover, this ML
fitting procedure provides an indirect estimate of the parameter
$\mu$. Since the right-tail of the Modified Inverse Gamma is a power
law with exponent $-(1+b)$ and Prediction~\ref{pred:1} suggests a
power law decay for $G$ with exponent $-(1+\mu)$, we conclude that
this implies that $\mu$ should be approximately $4.6$, a high value at
odds with the assumption that $1<\mu<2$. We return on this
inconsistency below.

Hence, Prediction~\ref{pred:1} of a size-independent universal
distribution for the rescaled volatilities is well obeyed by the
data. Furthermore, the distribution is well described by a Modified
Inverse Gamma distribution which displays a power law right-tail with
an exponent of about $-(1+4.6)$. These two results are new to the
literature, and represent a first empirical contribution of the
present work.  However, there is no clear sign of a hump for large
realizations in the distribution of the rescaled growth volatilities
$P(\sigma/\bar{\sigma}(S))$ which would signal the existence of poorly
diversified firms with few sub-units.

We check the robustness of these results with respect to three
potential threats. First we replace the mean absolute deviation with
the standard deviation as a proxy for the growth rate volatility. The
reason for this check being to use an estimator closer to its
theoretical counterpart although more sensible to extreme
realizations. We find further corroboration to our results: the
distributions of growth volatilities proxied with the standard
deviation, once rescaled by their average by size bin, collapse again
on a unique distribution well approximated by an Inverse Gamma with a
shape parameter of about $4.7$ (cfr. Figure~\ref{fig-app:pred1_sd} in
Appendix~\ref{app:main}). Second, attrition in a panel of firms over
such a long period of time may affect the quality of our estimates of
$\sigma_i$, obtained as the mean absolute deviation, in particular
considering firms with a low number of growth rates. To alleviate this
concern we look at the growth volatility distribution for firms with
at least $20$ growth rates. Figure~\ref{fig-app:pred1_20obs}, reported
in Appendix~\ref{app:20obs}, show that the message with this
alternative sample does not change. In particular the estimated
exponent for the power law tail of the function $G(x)$ remains too
high, being $-(1+4.6)$. As a third, and final check, since Compustat
mixes firms reporting their balance sheet with respect to different
fiscal years, we focus on firms closing their accounts in
December. Figure~\ref{fig-app:pred1_fyr_end_Dec} in
Appendix~\ref{app:endDec} shows that also for this sample everything
remains the same including the point estimate of the power law
exponent found equal to $-(1+4.5)$.

\begin{figure}[t]
  \caption{Size and growth volatility}
  \label{fig:pred2}
  \begin{center}
    \input{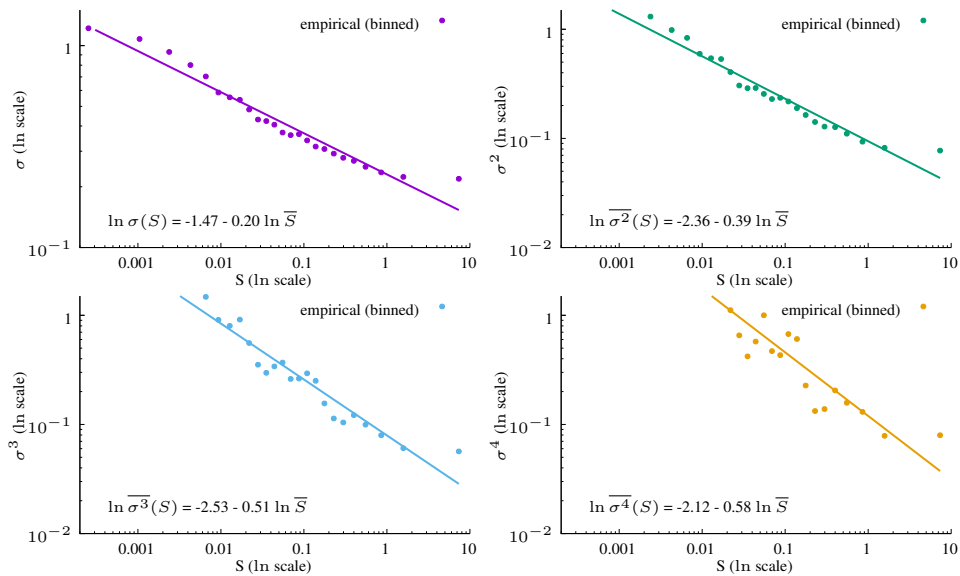}
  \end{center}
  {\scriptsize {\it Notes}. The four panels report on a double log
    scale the binned relation between normalized size and the first
    four sample moments of growth volatility together with an OLS
    linear fit. In all panel the number of bins is set to $25$ and
    volatility is computed as the mean absolute deviation multiplied
    by the factor $\sqrt{\pi/2}$. Data source: Compustat. \par}
\end{figure}
\medskip
 
We next study how the first four moments of the size-conditioned
volatility distribution scale with size. To guide this investigation,
we revert to Proposition~\ref{prop:vol_moments}, which leads to the
following prediction.
\begin{prediction}
  \label{pred:2}
  For large $S$ and when $1<\alpha<\mu+\frac{\mu-1}{\mu}$, the
  conditional mean growth volatility $\mathbb{E}[\sigma|S]$ decays with size as
  $S^{-\beta}$ with $\beta=\frac{\mu-1}{\mu}$. The 3 higher
  conditional moments of volatility, $\mathbb{E}[\sigma^q|S]$ with $q=2,3,4$,
  decay with size all with the same scaling, as $S^{\alpha-\mu}$.
\end{prediction}
To test this, we bin volatilities according to the same procedure
described above. Within each bin, we compute the average normalized
size $\bar{S}_{b}$ and the first 4 moments $\overline{\sigma^q}_b$
with $q=1,2,3,4$ an plot on a double log scale $\overline{\sigma^q}_b$
against $\bar{S}_{b}$. To help the visual inspection of the result, we
add the fit of the log-linear relation
$\log \overline{\sigma^q}_b = a_0 + a_1 \log \bar{S}_b$. The four panels of
Figure~\ref{fig:pred2} display the
results.

We begin with the top-left panel of Figure~\ref{fig:pred2}, reporting
the first moment of the growth volatility $\sigma$ conditioned on
size.\footnote{This is a replica of the right-panel of
  Figure~\ref{fig:motiv-evid}. We report again to clarify how we build
  it and to facilitate the comparison with the other 3 moments.} The
decay with size is approximately linear on a double log scale, meaning
that volatility decays as a power law. The only major deviation is
observed in the bin associated with the largest firms, which tend to
have a higher volatility than the one corresponding to the
scaling. These firms appear to be less effectively diversified than
what the model would lead us to expect.

The exponent is found to be about $-0.20(0.01)$, very much in line
with the existing literature. When the right tail of the distribution
of the number of sub-units $P(K)$ is heavy enough,\footnote{Heavy
  enough here means respecting the constraint
  $\alpha<\mu+\frac{1-\mu}{\mu}$.} this empirical result has an
important consequence. Indeed, in this case the estimated exponent has
been shown to be $\frac{1-\mu}{\mu}$ implying
$\mu=\frac{1}{1-\beta}=1.25$. This estimate is inconsistent with the
$4.5$ we have found when fitting the distribution of the rescaled
growth volatilities.\footnote{When the right tail of the distribution
  $P(K)$ is thin enough, $\alpha>\mu+\frac{1-\mu}{\mu}$, $E[\sigma|S]$
  would scale as $S^{\alpha-\mu}$ implying $\beta=\mu-\alpha=0.2$.}

This is confirmed by the study of higher order conditional
moments. Because $\alpha>1$, the theory predicts that all of them
should decay with size in the same way, as
$S^{\alpha-\mu}$. Unfortunately this is not what we observe in the
remaining three panels of Figure~\ref{pred:2}. While they all conform
to a power law decay, with possibly the same anomaly that the very
largest firms have higher values than expected, they have apparently
different scaling exponents, equal to $-0.39(0.02)$, $-0.51(0.03)$ and
$-0.58(0.04)$. This is in contradiction with Prediction~\ref{pred:2},
which states that their scaling should be {\it independent} of the
moment order $q$ and equal to $\alpha-\mu \approx -0.2$. Our evidence
suggests in contrast a clear dependence on the order $q$ of the
exponent. In fact, the scaling we obtain is close to the one we would
observe if the second term in Eq.~\eqref{eq:vol_moments}, proportional
to $S^{q\frac{1-\mu}{\mu}}$, was the dominant term, which is the
result expected for a {\it thin-tailed} distribution of rescaled
volatilities $\sigma/\overline{\sigma}(S)$, i.e. with a tail exponent
$b > \mu$, as indeed found above since $b \approx 4.6$. Such a
thin-tailed distribution of volatilities would asymptotically lead to
exponents $-0.4$, $-0.6$ and $-0.8$ for $q=2, \,3$ and $4$
respectively. We however expect subdominant corrections coming from
the incipient tail contribution to be responsible for the observed
discrepancies for $q \lesssim b$ ($-0.51$ instead of $-0.6$ for $q=3$
and $-0.58$ instead of $-0.8$ for $q=4$).\footnote{In fact, one
  predicts an asymptotic behaviour as $S^{- b \beta}$ independent of
  $q$ whenever $q > b$. } These results appear again robust when we
test them using the standard deviation to proxy growth volatilities
and different samples including only firms with at least 20 growth
rates and firms closing their fiscal year in December
(cfr. Appendix~\ref{app:f-inv}).

Summarizing these investigations on growth volatility, we may conclude
that although the granularity framework is partially consistent with
empirical evidence it also presents major drawbacks. We have indeed
shown that the rescaled volatility distribution is independent of size
and displays a power-law tail, being overall well approximated by a
Modified Inverse Gamma distribution with a thin right tail. However,
we do not see any sign of the non-scaling hump the model predicts for
large volatilities. Finally, the moments of size-conditioned growth
volatility distribution do decay with size as power laws, but with
exponents that are not wholly consistent with theoretical predictions.

\subsection{Growth rate distribution}
\begin{figure}[t]
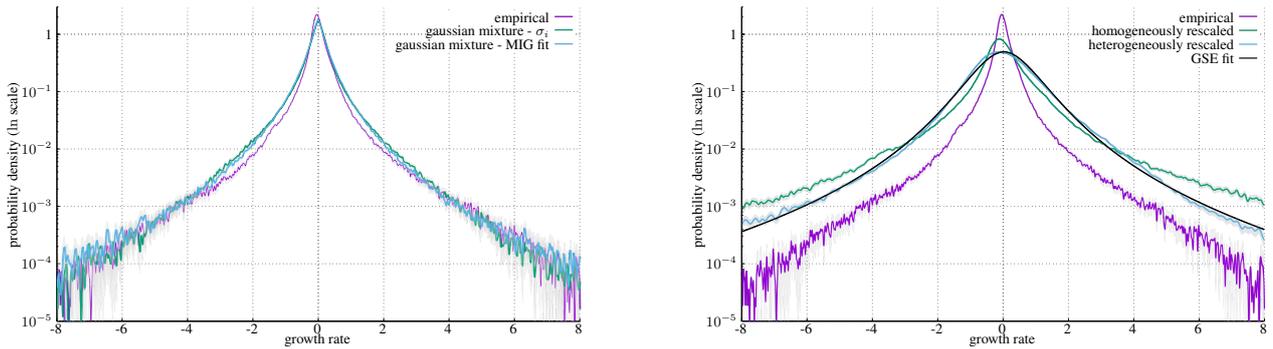

  \caption{Growth rate distribution}
  \label{fig:fgrd-shape}
  \begin{center}
    \begin{minipage}[t]{0.47\linewidth}
      \scalebox{0.6}{
      \input{./figures/fig_gr_norm_centered_mixtures.tex}
      }
    \end{minipage}
    \hfill
    \begin{minipage}[t]{0.47\linewidth}
      \scalebox{0.6}{
      \input{./figures/fig_gr_rescaled.tex}
      }
    \end{minipage}
  \end{center}
  {\scriptsize {\it Notes}: \underline{Left panel} reports the kernel
    estimate of the unconditional empirical density of the annual
    growth rate of the normalized quarter size after removing its mean
    and pooling growth rates across firms and year-quarter (purple
    line). We add the empirical densities estimated used synthetic
    growth rates obtained by mixing Gaussian random shocks with growth
    volatilities bootstrapped from those estimated as mean absolute
    deviation (green line) and sampling form the shifted Inverse Gamma
    distribution in Eq.~\eqref{eq:invgamma} (light-blue
    line). \underline{Right panel} compares the empirical benchmark
    (purple line) with the two empirical densities estimated on
    rescaled annual growth rates obtained by removing their mean and
    dividing by their volatility. First (green line) we use an
    homogeneous mean and an homogeneous volatility for the entire
    sample, second (light-blue line) we use individual specific means
    and volatilities computes as the leave-one-out mean absolute
    deviations. We add a GSE fit for the heterogeneously rescaled
    growth rates distribution computed on the kernel density estimate
    (black line). In all plots mean absolute deviations are
    adjusted by the factor $\sqrt{\pi/2}$. Kernel density estimates
    are computed with a Gaussian kernel function and a bandwidth
    chosen according to the normal reference rule-of-thumb. The
    estimate is evaluated on a 10,000 points regular grid defined over
    the entire empirical range. We also report the 95\% confidence
    interval. Data source: Compustat.\par}
\end{figure}

In this sub-section, our interest moves to the growth rates and their
distribution. Considering Proposition~\ref{prop:gmm_levy} is again not
trivial, since the use of the asymptotic representations in
Eq.~\eqref{eq:WB_GMM} and \eqref{eq:pred2_gr} requires to work in the
very large size region. In fact, it is only in this region that the Normal
approximation associated with the CLT is expected to be valid, and
where the Levy alpha-stable regime where $g\sim S^{-\beta}$ is well
separated from the non-universal regime where $g\sim 1$ and the growth
rate distribution depends on that of the sub-units growth
shocks. To make this even more challenging, our results indicate that
corrections to the CLT decay slowly, as $S^{1-\mu}$, and that the same
is true for the corrections associated with the existence of large firms with
only a handful of sub-units, which vanishes as
$S^{\alpha-\beta}$. Unfortunately, typical data-sets contain only a few
data points in the large size region, and so we expect finite
size corrections to play a significant role in shaping the empirical
growth rate distribution.

Despite this caveat, Proposition~\ref{prop:gmm_levy} remains effective
in offering guidance to look at growth rates and at their distribution
allowing us to make a sharp empirical prediction.
\begin{prediction}
  \label{pred:3}
  The unconditional growth rate distribution $P(g)$ can be
  rationalized as a mixture of Normal random variables, that is $P(g)$
  is approximately equal to
  $\int {\rm d}S {\rm d}\sigma \, \mathcal{N}(0,\sigma) \, P(\sigma \vert S)
  P(S)$. Net of finite size corrections, the growth rate distribution
  conditional on volatility and size, $P(g|\sigma,S)$, is Normal.
\end{prediction}

To test this prediction, we perform first a numerical exercise where
we simulate a growth rate distribution as a scale mixture of Normal
random variables. To this aim, we generate $1,000,000$ synthetic
growth rates as $g_i=\sigma_i \xi_i$ where $\xi_i$ is a Normal random
variable with zero mean and unit variance while $\sigma$ a random
variable with a density $P(\sigma)$.  The draws from the distribution
$P(\sigma)$ are done successively in two different ways. For the first
approach, we bootstrap $\sigma$ from the empirical observations. The
second approach assumes an Modified Inverse Gamma density for
$P(\sigma)$, with parameters fit using observed volatilities rescaled
by their average by size bin as in the right-panel of
Figure~\ref{fig:pred1}.

Our results are reported in the left-panel of
Figure~\ref{fig:fgrd-shape} (green and light-blue lines
respectively). Although some discrepancies are observed, the Gaussian
Mixture model generates a growth rate distribution in very good
agreement with empirical observation both in the central and tail
regions. The two procedures we have implemented produce qualitatively
equivalent results, and perform significantly better than the one
reported in Figure~\ref{fig:motiv-evid} (green line). The key
difference is that in Figure~\ref{fig:motiv-evid}, the volatilities
used to generate the artificial growth rates were simulated according
to $\sigma_i \sim S_i^{-\beta}$, thus neglecting the heterogeneity
associated with a firm's internal structure in terms of number and
size of sub-units. The resulting Gaussian mixture produces a
distribution that is not able to capture the central cusp and the fat
tails of its empirical target. With this respect a further remark is
in order. If, in the mixing of the Normal random variables, one uses
the MIG fitted on the non-rescaled volatilities the corresponding
growth rates display a distribution that, contrary to what happens
with rescaled volatilities, misses its empirical target. Indeed the
shape parameter $b$ of the MIG in this case is found to be
$2.35(0.03)$, significantly lower than $4.62(0.09)$. This implies a
fatter right tail for the volatility distribution inducing fatter
tails in the distribution of the growth rates generated as scale
mixture of Normals. We have clarified above why the estimate of $b$
with non-rescaled volatilities is likely to be biased due to size
heterogeneity. This is reassuring when it comes to the sensitivity and
consistency of the argument.

A second exercise, inspired by Prediction~\ref{pred:3}, concerns the
shape of the distribution of growth rates conditional on size and on
volatility, which is suggested to be approximately Normal. To test
this, we proceed as follows. For each firm in our sample, we compute
an average growth $\bar{g}_i$, and for each of its growth rates
$g_{it}$ we compute a growth volatility $\sigma_{it}$, computed as the
{\it leave-one-out} mean absolute deviation of the growth rates,
obtained from the time-series $g_{it'}$ where the growth rate $g_{it}$
for $t'=t$ has been discarded. We then define the corresponding
rescaled growth rates $\hat{g}_{it}$ as
$(g_{it}-\bar{g}_i)/\sigma_{it}$.  For consistency, we also apply the
{\it leave-one-out} procedure to the mean growth
$\bar{g}_i$.\footnote{The {\it leave-one-out} procedure avoids an
  artificial truncation of the tails, see
  e.g. \cite{bouchaud2003theory}.} The resulting distribution is
reported in the right-panel of Figure~\ref{fig:fgrd-shape} (light-blue
line). It is apparent that the rescaling procedure has a significant
effect on the shape of the growth rate distribution, smoothing the
central cusp and fattening both tails. These remain clearly
non-Gaussian and sub-exponential, since they display a positive
convexity in a linear-log representation. In the same figure, we
report for comparison the distribution we would obtain by rescaling
the growth rates with an homogeneous mean and an homogeneous
volatility for the entire sample (green line). This shows that the
associated reshaping, although still apparent, is significantly less
pronounced. It is equally apparent that the parabolic shape of a
Normal distribution would not provide a good fit for the entire
distribution.

\begin{table}[t]
  \centering 
  \caption{Symmetric Generalized Stretched Exponential fit} 
  \label{tab:gse_sym}
  \scalebox{0.75}{
    \begin{threeparttable}
\begin{tabular}{@{\extracolsep{5pt}} l cccccc} 
\\[-1.8ex]\hline 
\hline \\[-1.8ex] 
 & $C$ & $u$ & $v$ & $w$ & $z$ & Prob. mass $[-w,w]$ \\ 
\hline \\[-1.8ex] 
\underline{{\it Whole sample}} &   &   &   &   &   &  \\ 
homogeneous rescaling & 0.764 (0.004) & 0.432 (0.002) & 0.014 (0.000) & 0.673 (0.005) & 0.412 (0.001) & 0.684 \\ 
heterogeneous rescaling & 0.494 (0.003) & 0.863 (0.004) & 0.023 (0.001) & 1.777 (0.015) & 0.407 (0.003) & 0.897 \\ 
 &   &   &   &   &   &  \\
\end{tabular} 
      \begin{tablenotes}
        \footnotesize
      \item {\it Notes}: estimates are obtained with a non lienar
        least square algorithm applied to kernel density estimates
        evaluated on a regular grid in the interval $[-8,8]$. This
        procedure is applied to homogeneously rescaled growth rates,
        and heterogeneously rescaled growth rates in different size
        bins and computed over different time windows.  Data source:
        Compustat. \par
      \end{tablenotes}
    \end{threeparttable}
    }
\end{table}

We propose a phenomenological model to quantify the extent to which
this rescaling procedure move the growth rate distribution towards the
predicted Normal benchmark. We therefore consider a parametric
family of distributions, that we name Generalized Stretched
Exponential (GSE), with density
\begin{equation}
    \label{eq:gse-pdf}
    P_{GSE}(x;C,u,v,w,z)=C \exp\left( -\frac{(x-v)^2}{2u^2(1+(x/w)^{(2-z)})}\right)
\end{equation}
where $(u,v,w,z)$ are all positive parameters and $C$ is a
normalization constant. This family of distributions is built {\it
  ad-hoc} to feature a Gaussian behavior for small $x$ ($x < w$), and
a stretched ($z < 1$) or compressed ($z > 1$) exponential decay for
large $x$ ($x > w$).\footnote{More precisely when $x$ is large with
  respect to $w$ the tails of the density behave proportionally to
  $\exp\left[-\frac{x^z}{2u^2 w^z}\right]$.} The case $z=2$
corresponds to the standard Gaussian, whereas $z \to 0$ mimics
power-law tails.  To fit the GSE distribution we proceed as
follows. We first compute the kernel density estimate of the central
body of the empirical distribution of the heterogeneously rescaled
growth rates $\hat{g}_i$. Next, using a non linear least squares
algorithm, we fit the estimated density to recover the GSE parameters.

Figure~\ref{fig:fgrd-shape} reports, in the right-panel, this GSE fit
(black line) appearing to a large extent indistinguishable from the
empirical distribution, both in the central part and in the
tails. This confirms that the proposed family of distributions
performs very well in fitting the data. The estimated coefficients are
found equal to $C=0.483(0.004)$, $u=0.894(0.006)$, $v=-0.006(0.001)$,
$w=1.905(0.020)$ and $z=0.377(0.003)$. Such a fit allows us to provide
a quantitative assessment of the interval over which the distribution
can be considered approximately a Gaussian distribution. With an
estimated value of $w$ approximately equal to $1.905$ the empirical
distribution of $\hat{g}$ can be considered Gaussian in the range
$[-1.905,1.905]$. Further, we can compute by numerical integration of
the kernel density estimate, that more than $90\%$ of the probability
mass falls in this Gaussian regime. As a point of comparison, if one
were to consider in the same figure the homogeneously rescaled growth
rates (green line) $w$ estimate would reduce to $0.717$ and the
corresponding probability mass to $70\%$ corresponding to a smaller
Gaussian regime.

\begin{figure}[t]
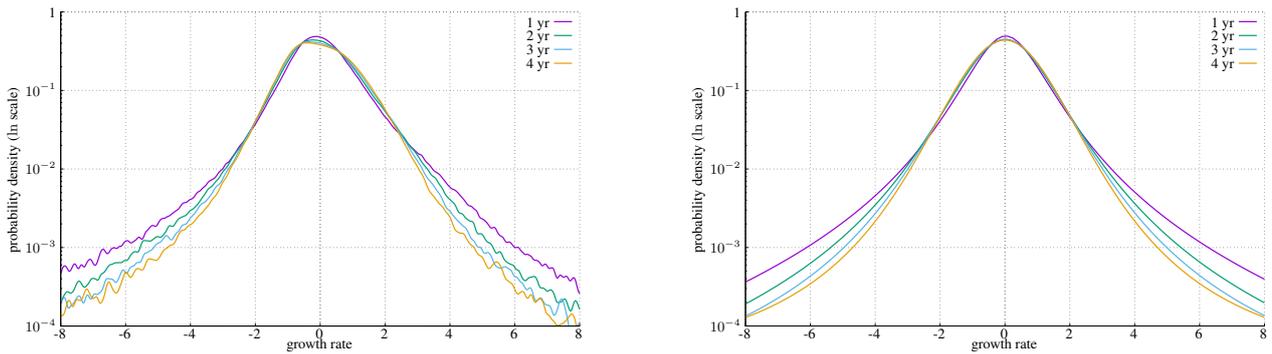

  \caption{Rescaled growth rate distribution (time windows)}
  \label{fig:fgrd-het-rescaled-dt}
  \begin{center}
    \begin{minipage}[t]{0.47\linewidth}
      \scalebox{0.6}{
        \input{./figures/fig_gr_rescaled_dt_kdens.tex}
      }
    \end{minipage}
    \hfill
    \begin{minipage}[t]{0.47\linewidth}
      \scalebox{0.6}{
        \input{./figures/fig_gr_rescaled_dt_gse.tex}
      }
    \end{minipage}
  \end{center}
  {\scriptsize {\it Notes}: \underline{Left panel} Kernel density
    estimates are computed with a Gaussian kernel function and a
    bandwidth chosen according to the normal reference
    rule-of-thumb. The estimate is evaluated on a 2,500 points regular
    grid defined over the interval $[-8,8]$. \underline{Right panel}
    Symmetric GSE fits computed with a non linear least square
    algorithm applied to density estimates. Tim-window over which
    growth rates are computed increases from 1 to 4 years moving from
    dark-violet, to green, blue and gold colors. Data source:
    Compustat.\par}
\end{figure}

However, the same GSE fit provides information on the non-Gaussian
nature of the two tails. Indeed the parameter $z$ is found to be
significantly smaller than one both for homogeneously and
heterogeneously rescaled growth rates identifying stretched
exponential tails. Hence tails fatter than those featured by a
Gaussian, a fact independently noted also in
\cite{larsen-hallock_etal_2022}.\footnote{In fact
  \cite{larsen-hallock_etal_2022} fit the tails with a power-law,
  which is, over a limited region, not very different from a stretched
  exponential with the observed low value of $z$.}

Taking stock of these findings, it looks fair to say that our data
provide a mixed and contrasting evidence in support to
Proposition~\ref{prop:gmm_levy}. On one side the Gaussian Mixture
model seems to work properly, at least in a first approximation, as an
explanation for the existence of a central cusp and of fat tails in
the empirical growth rate distribution. In line with the proposition
un-mixing the growth rates by rescaling them by the corresponding
individual mean and volatility, lead to an important extension of the
Gaussian regime in their empirical distribution. However the tails are
not absorbed within the Gaussian regime and remain fatter,
significantly departing from the strict Gaussian Mixture
model. Whether this is due to a flaw in the model rather than to
significant finite size corrections is what we attempt to investigate
with the last two exercises performed in the remaining of this
section.

As discussed in the theory section the quality of the Gaussian
approximation for the un-mixed growth rates is driven by two
factors. First, it is possible that we would observe large firms with
few sub-units, whose non Gaussian growth rates contaminate the Gaussian
mixture. Second, the very heterogeneous sizes of the sub-units which
make up a large firm slows down the convergence to the Normal limit
induced by the CLT. Driven by these considerations we perform two
investigations.

First, we exploit the fact that in a model where the growth rate
$g_i$, computed over a one year time window, is a random variable with
a GSE distribution multiplied by a firm specific volatility, then the
corresponding growth rates, computed over a larger time span, is
obtained as the sum of the random variables $g_i$. If these growth rates
are independent and identically distributed, the Central Limit Theorem
mechanism comes into play, which should lead to a progressive
``Gaussianization'' of the distribution of the corresponding un-mixed
growth rates $\hat{g}_i$. Figure~\ref{fig:fgrd-het-rescaled-dt}
shows that this is what happens empirically albeit with some
intriguing discrepancies. It reports in the left panel the kernel
density estimate of the distributions of the rescaled growth rate
computed over a 1, 2, 3 and 4 years time-span together with, in the
right panel, their GSE fit whose parameters estimates are shown in the
bottom panel of Table~\ref{tab:gse_sym}. In line with our reasoning,
we observe an apparent process of ``Gaussianization'' of the
distribution. The size of the Gaussian central region, measured by
$w$, almost double rising from $1.9$ to $3.5$ when we widen the time
window from 1 year to 4 years. Perhaps surprisingly, however, the
tails (described by the parameter $z$) appear to get {\it fatter} as
the time-span increases (see Table~\ref{tab:gse_sym}), although the
probability mass in that region significantly drops from $9\%$ at 1
year to less than $1\%$ at 4 years.  Correspondingly, when comparing
the empirical distribution for 2 years, 3 years and 4 years to a
numerical convolution of 2, 3 or 4 i.i.d.  random variables drawn from
the 1 year distribution, we find significant deviations: the tails of
empirical distributions are ``too fat'', signalling the existence of
persistence effects. In other words, the convolution results suggests
that growth rates are autocorrelated even over three and four years,
i.e. persistently good or bad periods or persistently volatile periods
build up tails that would not exist if the growth process was
independent across time.

\begin{figure}[t]
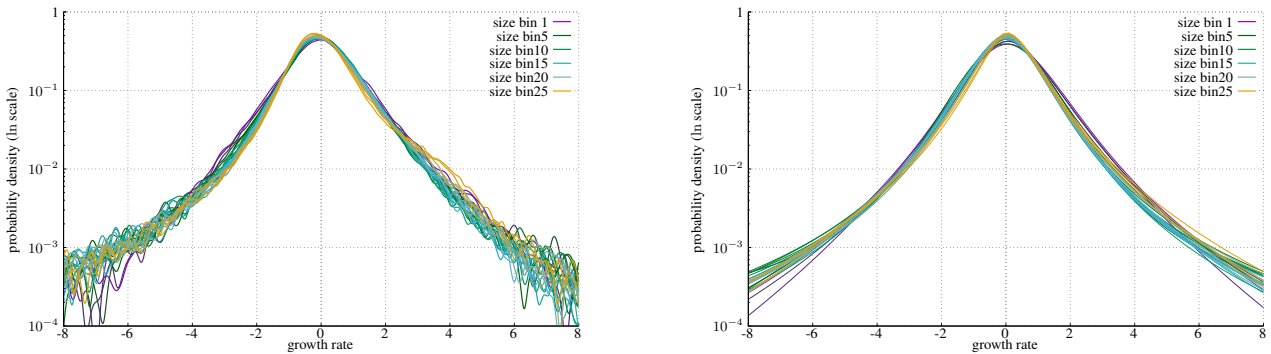

  \caption{Rescaled growth rate distribution (size bins)}
  \label{fig:fgrd-het-rescaled}
  \begin{center}
    \begin{minipage}[t]{0.47\linewidth}
      \scalebox{0.6}{
        \input{./figures/fig_gr_rescaled_bin25_kdens_firm.tex}
      }
    \end{minipage}
    \hfill
    \begin{minipage}[t]{0.47\linewidth}
      \scalebox{0.6}{
        \input{./figures/fig_gr_rescaled_bin25_gse_firm.tex}
      }
    \end{minipage}
  \end{center}
  {\scriptsize {\it Notes}: \underline{Left panel} Kernel density
    estimates are computed with a Gaussian kernel function and a
    bandwidth chosen according to the normal reference
    rule-of-thumb. The estimate is evaluated on a 2,500 points regular
    grid defined over the interval $[-8,8]$. \underline{Right panel}
    Symmetric GSE fits computed with a non linear least square
    algorithm applied to density estimates. Bin size increases moving
    from dark-violet, to green, blue and gold colors. Data source:
    Compustat.\par}
\end{figure}

In a second, final exercise we investigate the size dependency of the
rescaled growth rates distribution. To this aim we bin firms in 25
groups according to their average size $\bar{S}_i$ following the same
procedure described in the data section and used to investigate the
growth volatility distribution in
Section~\ref{sec:gr-volatility}. Figure~\ref{fig:fgrd-het-rescaled}
reports in the left panel the kernel density estimate of the
distributions for the 25 bins together with, in the right panel, their
GSE fit whose parameters estimates are shown in the middle panel of
Table~\ref{tab:gse_sym}. Visually all distributions fall approximately
on one another, suggesting that, once un-mixed using idiosyncratic
means and volatilities, the distribution of growth rate looks
universal and independent of size.

In other words we see {\it no sign} of ``Gaussianization'' for larger
sizes, contrarily to what we were expecting. Such visual inspection is
confirmed by the parameters estimates reported in
Table~\ref{tab:gse_sym} where we do not observe any clear trend in any
of the key parameters. If anything the parameter $w$ defining the
Gaussian region actually {\it decreases} with firm average size
$\bar{S}$, together with the probability mass in the interval $[-w,w]$
which drops from approximately $0.95\%$ for small firms to $0.84$ for
the largest firms.

This observation means that rescaled growth rates of small firms are
{\it more Gaussian} than those of large firms, which goes against the
intuition provided by Prediction~\ref{pred:3}. This result may be
related to the fact that the exponent $\beta$ is smaller for large
firms, indicating that -- somewhat paradoxically -- the activity of
large firms is more volatile than one might have anticipated from the
$S^{-\beta}$ scaling. Similarly $z$ does not show a clear trend with
size, remaining for all bins well below the threshold identifying
stretched exponential tails and showing a rather good deal of
homogeneity across bins and with respect to the results obtained
pooling together all firms.

\section{Conclusions}
\label{sec:conclusion}

This paper shows that granular models of firm growth, as those in
\cite{wyart2003statistical} and \cite{Gabaix_2011}, do not provide a
convincing explanation of the empirically robust scaling of the
growth volatility with firm size, $\sigma \sim S^{-\beta}$, with an exponent $\beta$ lower
than $0.5$. Nor do they account, as a consequence, for the shape of the distribution of the growth
rates of business firms.

The basic ``granular'' assumption rationalizes such anomalous scaling
by the concentration of the sales of firms across sub-units which
would be dominated by only few of them representing for example
blockbuster products or large  customers, as in \cite{herskovic2020firm}. Although enticing, such a picture leads to
falsifiable predictions, in particular concerning the distribution of
volatility conditional to size, for which we provided new theoretical
results. Quite remarkably, and as predicted by theory, this
distribution is indeed close to universal (i.e., independent of firm size) once rescaled by the average
size dependent volatility $S^{-\beta}$. However, the tail of this universal
distribution is much too ``thin'' to be compatible with the
granularity hypothesis. 

Similarly, we do {\it not} observe the
anomalous scaling of higher moments of the growth rate volatility with
firm size that should reflect the presence of a cohort of very poorly
diversified firms. The data unambiguously rejects this prediction. At
the same time, and again contrary to what is predicted by models,
after accounting for idiosyncratic volatilities growth rate shocks
remain non-Gaussian and fat tailed. Correspondingly, high moments of
growth rates but {\it not} of their volatility are dominated by such
tail events, a result also noted in
\cite{larsen-hallock_etal_2022}. 

These findings reveal more unanswered
questions. Where do such very large idiosyncratic shocks come from? Why
are large firms even {\it more} volatile than expected, and with more
weight in the non-Gaussian tails? While this paper shows that the answers to these questions cannot be searched for in granular mechanisms alone it is not aimed at proposing alternative explanations. One intriguing possibility is that shocks hitting a firm's sub-units become more and more correlated as its size increases, due either to supply chain effects affecting inputs or customers common to all sub-units,  or to reputation shocks that impact the whole firm. Although such an analysis goes beyond the scope of the present paper, we have explored this hypothesis finding encouraging evidence: both the correlation of a firm's growth rate with that of the whole economy and the correlation of the growth rate of firms of similar sizes clearly {\it grow} with their size. Very much along the lines of  \cite{herskovic2020firm}, we surmise that this feature is the key missing ingredient to fully account for the observed statistics of firm growth.

\section*{Acknowledgments} We want to thank E. Larsen-Hallock, A. Rej
and D. Thesmar for sharing their results with us and for
discussions. JM wishes to thank M. Benzaquen, D. Farmer, A. Fosset,
F. Lafond, L. Mungo, M. Tomas and M. Riccaboni for numerous
discussions on the topic of firm growth modeling. We also thank
participants at the 2023 Computations in Economics and Finance
conference (Nice) for helpful comments.

\newpage

\bibliography{biblio.bib}

\newpage
\appendix

\footnotesize

\noindent{\Large {\bf Appendixes}}

\section{Analytical results and proofs}
\label{app:proofs}

\section*{Single granularity hypothesis [Gabaix' model]}
\subsection{Preliminary results}

We begin by considering a single firm consisting of a deterministic
number of sub-units $K$ of sizes $\left(s_1,\ldots,s_K\right)$
distributed according to Assumption~\ref{hp:su_size}. With the aim of
studying the statistics of the firm's Herfindahl-Hirschman index
$\mathcal{H}$ of the sub-unit size we will first establish the scaling
of the typical (modal) value of $\mathcal{H}$ by estimating
the numerator and the denominator of $\mathcal{H}$, before computing
its moments explicitly and deriving its whole distribution. Without
loss of generality, we will take $s_0 = 1$. Indeed, under the current
hypotheses it is easy to see that
$\mathbb{E}[s]=\frac{\mu}{\mu-1}s_0$, and so taking $s_0=1$ amounts to
redefining the units in which sub-unit sizes are counted.  \medskip

The following lemmas establish three preliminary results.
\begin{lemma}
  \label{lem:typical_herfindahl}
  The typical value of the Herfindahl index scales with $K$ as
  \begin{equation}
    \label{eq:typical_herfindahl}
    \mathcal{H}_{\text{typ}} \sim K^{2\frac{1-\mu}{\mu}},\quad K\gg 1.
  \end{equation}
\end{lemma}
\begin{proof}
  Starting from the definition
  $\mathcal{H}=\sum_j (s_{ji}/\sum_h s_{hi})^2$, note that by the law
  of large numbers since $\mathbb{E}[s]<\infty$ we have that
  $\sum_j s_j \approx K \mathbb{E}[s]$. On the other hand, the
  expectation of the numerator $\sum_i s_i^2$ is divergent, since
  $\mathbb{E}[s^2]=\infty$.

  However, the largest element $s_{\max}$ of
  $\left(s_i,\ldots, s_k\right)$ scales as
  $s_{\max}\sim K^{\frac{1}{\mu}}$, since it solves

  \begin{equation}\label{eq:typical_max}
    \int_{s_{\max}}^\infty \mathrm{d} s~P(s) \approx \frac{1}{K}, 
\end{equation}
leading to $s_{\max}^{-\mu}\approx K^{-1}$. The sum $\sum_i s_i^2$ can
be approximated by $s_{\max}^2\approx K^{\frac{2}{\mu}}$, which
diverges superlinearly as $K\to\infty$. Pooling these results
together, we obtain
\begin{equation}
  \mathcal{H}_{\text{typ}} \sim \frac{K^{\frac{2}{\mu}}}{K^2\mathbb{E}[s]^2}\sim K^{2\frac{1-\mu}{\mu}}.
\end{equation}
\end{proof}

\begin{lemma}
  \label{lem:power_law_laplace}
  The moment generating function of sub-unit sizes
  $g(\lambda)= \mathbb{E}[e^{-\lambda s}]$ has an expansion
  \begin{equation}\label{eq:Laplace_pl_expansion}
    g(\lambda) \approx 1-\lambda \mathbb{E}[s]-\lambda^{\mu}\mu\Gamma[-\mu]+\mathcal{O}\left(\lambda^{\mu}\right),
  \end{equation}
  which corresponds to the Laplace transform of a probability density
  decaying as $s^{-1-\mu}$.
\end{lemma}
\begin{proof}
  This extends results obtained
  by~\cite{derridaRandomWalksSpin1997}. By direct computation,
  \begin{align}
    \label{eq-app:pl_laplace1}
    g(\lambda)=\int_{s_0}^\infty\dint s ~P(s)\exp(-\lambda s) &=  \int_{s_0}^\infty\dint s \dfrac{\mu}{s^{1+\mu}} \exp(-\lambda s) \nonumber\\
    &\text{[setting $s_0=1$ and changing variable $\lambda s=t$]}\nonumber\\
    &=\mu \lambda^\mu\int_{\lambda}^\infty\dint t~ t^{-\mu-1} \exp(-t)\nonumber\\
    &=\mu \lambda^\mu \Gamma(-\mu,\lambda)=\mu \lambda^\mu \left[\Gamma\left[-\mu\right] - \gamma(-\mu,\lambda) \right]\nonumber\\
    &\text{[$\Gamma(\cdot,\cdot)$ and $\gamma(\cdot,\cdot)$   are the upper and lower incomplete gamma functions]}\nonumber\\
    &=-\lambda^\mu \Gamma\left[-\mu+1\right] - \mu\lambda^\mu\left[\lambda^{-\mu} \sum_{k=0}^{\infty} \frac{(-\lambda)^{k}}{k!(-\mu+k)} \right] \nonumber\\
    &=-\lambda^\mu \Gamma\left[-\mu+1\right]+1-\lambda\frac{\mu}{\mu-1} + \mathcal{O}(\lambda^{\mu}) \nonumber\\
    &=1-\lambda^\mu \Gamma\left[-\mu+1\right]-\lambda \mathbb{E}[s] + \mathcal{O}(\lambda^{\mu})  \;\;.
  \end{align}
\end{proof}

A useful corollary follows.
\begin{corollary}
  \label{cor:laplace_moments}
  The expansion for $\mathbb{E}[s^k e^{-\lambda s}]$ reads
  \begin{equation}
    \mathbb{E}[s^k e^{-\lambda s}] \approx \frac{\mu}{\mu-k}+\mu \Gamma[k-\mu]\lambda^{\mu-k}+\frac{\mu\lambda}{1-\mu+k}+\mathcal{O}\left(\lambda^{\min(1,\mu-k)}\right).
  \end{equation}
\end{corollary}
\begin{proof}
  Note that
  \begin{equation}
    \begin{split}
      \frac{\dint ^k g(\lambda)}{\dint \lambda^k }  &= \mathbb{E}\left[\frac{\dint^k}{\dint \lambda^k}[e^{-\lambda s}]\right]\\
      &= \mathbb{E}\left[(-s)^k e^{-\lambda s}\right], 
    \end{split}
  \end{equation}
  so that $\mathbb{E}[s^k e^{-\lambda s}] = (-1)^k
  g^{(k)}(\lambda)$. Differentiating the expression in the second to
  last line of Eq.~\eqref{eq-app:pl_laplace1} yields the result.
\end{proof}

\begin{lemma}
  \label{lem:integer_herfindahl_moments}
  
  The integer moments of the Herfindahl read
  \begin{equation}
    \label{eq:herfindahl_moments}
    \mathbb{E}\left[\mathcal{H}^k\right] \approx K^{1-\mu}\mathbb{E}[s]^{-\mu}\mu \frac{\Gamma[\mu]\Gamma[2k-\mu]}{\Gamma[2k]}+\mathcal{O}(K^{1-\mu})
  \end{equation}
  for $k\geq 1$. In particular, the average Herfindahl scales as
  $\mathbb{E}[\mathcal{H}]\sim N^{1-\mu}$.
\end{lemma}
\begin{proof}
We use the following representation:
  
  \begin{equation}
    \label{eq:1x_identity}
    \frac{1}{x^k} = \frac{1}{\Gamma\left[k\right]}\int_{0}^{^\infty}\dint \lambda~\lambda^{k-1} e^{-\lambda x}.
  \end{equation}
  Indeed, changing variables first as $t=\lambda x$ gives
  $\int_{0}^{^\infty}\dint \lambda~\lambda^{k-1} e^{-\lambda x} =
  x^{-k} \int_0^\infty \dint t~
  t^{k-1}e^{-t}=x^{-k}\Gamma\left[k\right]$.
  
  Now write
  $\mathcal{H}^k = \left(\sum_{i_1}\ldots \sum_{i_k} s_{i_1}^2\ldots
    s_{i_k}^2\right)/\left(\sum_i s_i\right)^{2k}$, so that using the
  identity in Eq.~\eqref{eq:1x_identity} we have that
  \begin{equation}
    \label{eq:herfindahl_rewriting_1}
    \begin{split}
      \mathbb{E}\left[\mathcal{H}^k \right] &= \mathbb{E}\left[\frac{1}{\Gamma[2k]}\int_{0}^{^\infty}\dint \lambda~\lambda^{2k-1} e^{-\lambda \sum_i s_i}\sum_{i_1}\ldots \sum_{i_k} s_{i_1}^2\ldots s_{i_k}^2\right].
    \end{split}
  \end{equation}
  We will now work out the case $k=1$ before generalising for
  $k\geq 1$. Eq.~\eqref{eq:herfindahl_rewriting_1} becomes
  \begin{equation}
    \label{eq:herfindahl_rew_2}
    \begin{split}
      \mathbb{E}\left[\mathcal{H}\right] &= \mathbb{E}\left[\int_{0}^{^\infty}\dint \lambda~\lambda e^{-\lambda \sum_i s_i}\sum_{j} s_{j}^2\right]\\
      &= \int_{0}^{^\infty}\dint \lambda~\lambda \mathbb{E}\left[e^{-\lambda s_1 -\lambda \sum_{i>1}s_i}\left(s_1^2 + \sum_{j>1}s_j^2\right)\right]\\
      &=  \int_{0}^{^\infty}\dint \lambda~\lambda \mathbb{E}\left[s_1^2e^{-\lambda s_1}\right] \mathbb{E}\left[e^{-\lambda \sum_{i>1}s_i}\right]+\int_0^\infty \dint \lambda~\lambda \mathbb{E}\left[e^{-\lambda \sum_{i>1}s_i}\sum_{j>1}s_j^2\right]\\
      &= K \int_{0}^\infty \dint \lambda~\lambda \mathbb{E}\left[s^2 e^{-\lambda s}\right]\mathbb{E}\left[e^{-\lambda s}\right]^{K-1}.
    \end{split}
  \end{equation}
  We now use the results of Lemma~\ref{lem:power_law_laplace} and
  Corollary~\ref{cor:laplace_moments}, so that to leading order
  \begin{equation}
    \mathbb{E}\left[e^{-\lambda s}\right]^{N-1} \approx \left(1-\lambda \mathbb{E}[s]\right)^{K-1}\approx e^{-\lambda \mathbb{E}[s]K},   
  \end{equation} 
  and
  \begin{equation}
    \mathbb{E}\left[s^2 e^{-\lambda s}\right] \approx \mu \Gamma\left[2-\mu\right]\lambda^{\mu-2}.
  \end{equation}
  Plugging into the integral, we obtain
  \begin{equation}
    \label{eq:H_avg_final}
    \begin{split}
      \mathbb{E}\left[\mathcal{H}\right] \approx& \mu \Gamma\left[2-\mu\right] K \int_{0}^\infty \dint \lambda~\lambda \exp\left(-\lambda K \mathbb{E}[s]\right)\lambda^{\mu-2}\\
      \approx& \mu \Gamma\left[2-\mu\right]K \int_{0}^{\infty} \dint \lambda~ \lambda^{\mu-1}\exp\left(-\lambda K \mathbb{E}[s]\right)\\
      &\text{[changing variables $t=\lambda K \mathbb{E}[s]$]}\\
      \approx& \mu \Gamma\left[2-\mu\right]K^{1-\mu}\mathbb{E}[s]^{-\mu}\int_0^\infty \dint t~t^{\mu-1}e^{-\lambda t}\\
      \approx & K^{1-\mu}\mu \Gamma\left[2-\mu\right]\Gamma[\mu] \mathbb{E}[s]^{-\mu},
    \end{split}
  \end{equation}
  so that $\mathbb{E}[\mathcal{H}]\sim K^{1-\mu}$. In the general
  case, we may proceed in similar fashion,
  \begin{equation}
    \label{eq:gen_k_case}
    \begin{split}
      \mathbb{E}\left[e^{-\lambda \sum_i s_i}\sum_{i_1}\ldots \sum_{i_k} s_{i_1}^2\ldots s_{i_k}^2\right] =& \mathbb{E}\left[\left(e^{-\lambda s_1 - \lambda \sum_{i>1}s_i}\right)\left(s_1^{2k}+k s_1^{2k-2}\sum_{j>1}s_j^2 +\ldots\right)\right]\\
      =& K\mathbb{E}\left[s^{2k}e^{-\lambda s}\right]\mathbb{E}\left[e^{-\lambda s}\right]^{K-1}+k(K-1)\mathbb{E}\left[s^{2k-2}e^{-\lambda s}\right]\mathbb{E}[s^2 e^{-\lambda s}]\mathbb{E}[e^{-\lambda s}]^{K-2}+\ldots,
    \end{split}
  \end{equation}
  so that we may split the different contributions as 
  \begin{equation}
    \label{eq:H_split}
    \begin{split}
      \mathbb{E}\left[\mathcal{H}^k\right]&= \frac{1}{\Gamma[2k]}\int_0^\infty \dint \lambda~\lambda^{2k-1}\left(K\mathbb{E}\left[s^{2k}e^{-\lambda s}\right]\mathbb{E}\left[e^{-\lambda s}\right]^{K-1}+k(K-1)\mathbb{E}\left[s^{2k-2}e^{-\lambda s}\right]\mathbb{E}[s^2 e^{-\lambda s}]\mathbb{E}[e^{-\lambda s}]^{K-2}+\ldots\right)\\
      &= I_1 + I_2+\ldots,
    \end{split}
  \end{equation}
  where $I_\ell$ is the integral with $j_1=j_2=\ldots=j_{\ell}=1$ in
  the sum of Eq.~\eqref{eq:gen_k_case}.  We next have to leading order
  in $\lambda$
  \begin{equation}
    \mathbb{E}\left[s^{2k}e^{-\lambda s}\right]\mathbb{E}\left[e^{-\lambda s}\right]^{K-1} \approx \mu\Gamma[2k-\mu]e^{-\lambda \mathbb{E}[s]K}\lambda^{\mu-2k},
  \end{equation} 
  so that
  \begin{equation}
    I_1 \approx \frac{K\mu \Gamma[2k-\mu]}{\Gamma[2k]}\int_{0}^\infty \dint \lambda~\lambda^{\mu-1}e^{-\lambda \mathbb{E}[s]K} \approx \mu K^{1-\mu}\mathbb{E}[s]^{-\mu}\frac{ \Gamma[2k-\mu]\Gamma[\mu]}{\Gamma[2k]},
  \end{equation}
  with subleading corrections that are
  $\mathcal{O}\left(K^{1-\mu}\right)$. Regarding $I_2$, we have that
  \begin{equation}
    \mathbb{E}\left[s^{2k-2}e^{-\lambda s}\right]\mathbb{E}[s^2 e^{-\lambda s}]\mathbb{E}[e^{-\lambda s}]^{K-2} \approx \mu^2\Gamma[2k-2-\mu]\Gamma[2-\mu] e^{-\lambda \mathbb{E}[s]K}\lambda^{\mu-2k-2}\lambda^{\mu-2}\propto \lambda^{2\mu-2k},
  \end{equation}
  so that $I_2 \sim K^{1-2\mu}=\mathcal{O}(K^{1-\mu})$. In the end, the
  dominating contribution is given by $I_1$ for $k\geq 1$, and so
  finally
  \begin{equation}
    \mathbb{E}[\mathcal{H}^k] =  K^{1-\mu} \mu \mathbb{E}[s]^{-\mu}\frac{ \Gamma[2k-\mu]\Gamma[\mu]}{\Gamma[2k]}+\mathcal{O}\left(K^{1-\mu}\right).
  \end{equation}
\end{proof}

\subsection{Proof of Proposition~\ref{prop:herfindahl_gabaix}.}
\label{ssec:proof_prop_gabaix1}

\setcounter{proposition}{0}
\begin{proposition}
  For fixed $K\gg1$, the distribution of the Herfindahl-Hirschman
  index takes the following form,
  \begin{equation}
    P(\mathcal{H}\vert K) = K^{2(\mu-1)/\mu} F\left(\mathcal{H}K^{2(\mu-1)/\mu}\right)\left(1-\sqrt{\mathcal{H}}\right)^{\mu-1}\;\;\;, 
  \end{equation}
  with $F(x) \sim \frac{1}{x^{1+\mu/2}}$ when $1 \ll x \lesssim K^{2(\mu-1)/\mu}$.
\end{proposition}
\begin{proof}
  Following what we obtained previously, we recognise the Beta
  function
  $\text{B}(a,b)=\frac{\Gamma[a]\Gamma[b]}{\Gamma[a+b]}=\int_0^1 \dint
  x~x^{a-1}(1-x)^{b-1}$ in Eq.~\eqref{eq:herfindahl_moments} and write
  \begin{equation}\label{eq:moment_identification}
    \begin{split}
      \frac{ \Gamma[2k-\mu]\Gamma[\mu]}{\Gamma[2k]} =& \int_0^1\dint x~x^{2k-\mu-1}\left(1-x\right)^{\mu-1}\\
      &\text{[setting $h = \sqrt{x}$]}\\
      &= \int_{0}^1 \dint h ~ h^{k-1-\frac{\mu}{2}}\left(1-\sqrt{h}\right)^{\mu-1}.
    \end{split}
  \end{equation}
  However, because we can identify
  $\mathbb{E}\left[\mathcal{H}^k\right]=\int_{0}^1 \dint
  \mathcal{H}~P(\mathcal{H})\mathcal{H}^k$, it follows that we can
  identify the leading contribution to this integral, coming from the
  region $\mathcal{H}\approx 1$,
  \begin{equation}
    P(\mathcal{H})\sim \mathcal{H}^{-1-\frac{\mu}{2}}\left(1-\sqrt{\mathcal{H}}\right)^{\mu-1}.
  \end{equation}
  Supplementing this with the fact that the typical value of
  $\mathcal{H}$ should be
  $\mathcal{H}_{\text{typ}}\sim K^{2\frac{1-\mu}{\mu}}$, this supposes
  that for $\mathcal{H}\ll 1$ we should have a scaling form,
  \begin{equation}
    P(\mathcal{H}) \approx \frac{1}{\mathcal{H}_{\text{typ}}}F\left(\frac{\mathcal{H}}{\mathcal{H}_{\text{typ}}}\right),
  \end{equation}
  that is such that $F(x)\sim x^{-1-\mu/2}$ for large
  $x\sim \frac{1}{\mathcal{H}_\text{typ}}$, in consistency with the
  form given in Eq.~\eqref{eq:moment_identification}. Furthermore,
  this is consistent with the fact that for large $K$, the typical
  Herfindahl reads
  \begin{equation}
    \mathcal{H} \approx \frac{1}{K\mathbb{E}[s]}\sum_i s_i^2,
  \end{equation} 
  and so the generalised law of large numbers indicates that
  $\sum_i s_i^2$ will have a distribution that is a Lévy-stable
  distribution with parameter $\mu/2$. The factor
  $\left(1-\sqrt{\mathcal{H}}\right)^{1-\mu}$ reconciles this with the fact that
  the Herfindahl is bounded by $1$ and with the fact that entities
  with $\mathcal{\mathcal{H}}\approx 1$ are statistically unlikely. Finally note
  that writing $\sigma \propto \sqrt{\mathcal{H}}$ leads to
  Eq.~\eqref{eq:Herf2}.
\end{proof}

\subsection{Proof of Proposition~\ref{prop:herfindahl_gabaix_moments}.}
\label{ssec:proof_prop_gabaix2}

\begin{proposition}
  For all $q>0$, the moments of $\mathcal{H}$ conditional on $K$ are
  given, to the leading order, by
  \begin{equation}
    \mathbb{E}\left[\mathcal{H}^q|K \right]\approx C_1 K^{1-\mu}+C_2K^{2q\frac{1-\mu}{\mu}} +\mathcal{O}\left(K^{\min\left(1-\mu, 2q \frac{1-\mu}{\mu}\right)}\right),
  \end{equation}
  where $C_1$ and $C_2$ are two constants. 
\end{proposition}
\begin{proof}
  We write the generalised moment as
  \begin{equation}
    \begin{split}
      \mathbb{E}\left[\mathcal{H}^q\right] &= \int_0^1\frac{\dint \mathcal{H}}{\mathcal{H}_\text{typ}} F\left(\frac{\mathcal{H}}{\mathcal{H}_{\text{typ}}}\right)\left(1-\sqrt{\mathcal{H}}\right)^{\mu-1}\\
      &= \mathcal{H}_{\text{typ}}^q \int_0^{1/\mathcal{H}_{\text{typ}}}\dint x~x^{q} F(x)\left(1-\sqrt{x\mathcal{H}_{\text{typ}}}\right)^{\mu-1}.
    \end{split}
  \end{equation}
  We may then approximate this integral using the intermediate tail of
  $F$, noting also that $F$ introduces a lower cut-off to the
  integral. The integral therefore runs from
  $\mathcal{H}\approx \mathcal{H}_{\text{typ}}$ to
  $\mathcal{H}\approx 1$. Therefore,
  \begin{equation}\
    \label{eq:gen_moments_2}
    \begin{split}
      \mathbb{E}\left[\mathcal{H}^q\right] \sim \mathcal{H}_{\text{typ}}^q \int_{1}^{1/\mathcal{H}_{\text{typ}}}\dint x~x^{q-1-\mu/2}\left(1-\sqrt{x\mathcal{H}_{\text{typ}}}\right)^{\mu-1}.
    \end{split}
  \end{equation}
  Next, we introduce $u=\sqrt{x\mathcal{H}_{\text{typ}}}$, with
  $\dint x = 2\frac{u}{\mathcal{H}_{\text{typ}}}\dint u$ and the
  integral runs from $u=\sqrt{\mathcal{H}_{\text{typ}}}$ to
  $u=1$. Therefore,
  \begin{equation}\label{eq_app:h_change_vars}
    \begin{split}
      \mathbb{E}\left[\mathcal{H}^q\right] \sim& 2\mathcal{H}_{\text{typ}}^q \frac{\mathcal{H}_{\text{typ}}^{1+\mu/2-q}}{\mathcal{H}_{\text{typ}}}\int_{\sqrt{\mathcal{H}_{\text{typ}}}}^1 \dint u~ u^{2q-\mu-1}\left(1-u\right)^{\mu-1}\\
      \sim& 2\mathcal{H}_{\text{typ}}^{\mu/2}\left(\int_{0}^1 \dint u~ u^{2q-\mu-1}\left(1-u\right)^{\mu-1}-\int_0^{\sqrt{\mathcal{H}_{\text{typ}}}} \dint u~ u^{2q-\mu-1}\left(1-u\right)^{\mu-1}\right)\\
      \sim& 2\mathcal{H}_{\text{typ}}^{\mu/2}\text{B}\left(2q-\mu,\mu\right)-2\mathcal{H}_{\text{typ}}^{\mu/2} \text{B}_{\sqrt{\mathcal{H}_{\text{typ}}}}\left(2q-\mu,\mu\right),
    \end{split}
  \end{equation}
  where $\text{B}(a,b)$ is the Beta function described above, and
  $\text{B}_x(a,b)=\int_0^x\dint t~t^{a-1}\left(1-t\right)^{b-1}$ is
  the incomplete Beta function. Using then that
  $B_{x}(a,b)= \frac{x^a}{a}\left(1+\mathcal{O}(x)\right)$, we have
  finally that
  \begin{equation}
    \label{eq:expansion_h_alpha}
    \mathbb{E}\left[\mathcal{H}^q\right] \sim  \text{B}\left(2q-\mu,\mu\right)\mathcal{H}_{\text{typ}}^{\mu/2}+\frac{1}{\mu-2q}\mathcal{H}_{\text{typ}}^q,
  \end{equation}
  which, after plugging in
  $\mathcal{H}_{\text{typ}}\sim K^{2\frac{1-\mu}{\mu}}$, yields the
  result given above. In addition to this, it becomes clear from the
  computation that the contribution $\mathcal{H}_{\text{typ}}^q$ comes
  from the values close to $\mathcal{H}\approx \mathcal{H}_{typ}$,
  while the contribution $\mathcal{H}_{\text{typ}}^{\mu/2}$ comes from
  the tail.
\end{proof}

\section*{Double granularity hypothesis [Wyart and Bouchaud's model]}
\subsection{Preliminary results}

The statistical object of interest in this set-up is
$P(S,R)$,\footnote{To ease the notation, we drop the firm index $i$
  and the time index $t$.}  the joint distribution describing firms of
size $S$ that experience an absolute size change $R$ between $t$ and
$t+1$.  Formally, we may write this as
\begin{equation}
  \label{eq:PSR_def}
  P(S,R) = \sum_{K>1}p(K)\int_{0}^{^\infty}\prod_{i=1}^K\dint
  s_i~p(s_i)\delta\left(S-\sum_{j=1}^K s_j\right)\int \prod_{j=1}^K\dint
  \eta_j~p(\eta_j)\delta\left(R-\sum_{j=1}^K \eta_j s_j\right),
\end{equation}
where $\delta(\cdot)$ is a Dirac delta distribution.
Equation~\eqref{eq:PSR_def} represents $P(S,R)$ as the sum over all
possible arrangements of Gaussian shocks, sub-unit numbers and
sub-unit sizes that result in an initial size equal to $S$ and to an
absolute growth equal to $R$. The role of the Dirac delta is to impose
the constraints $S=\sum_{j}s_j$ and $R=\sum_j \eta_j s_j$, and gives
this representation the structure of a sum of convolutions over the
different variables of interest.

Next, we ground the study of this object on on three pillars.
We first exploit the convolutional structure of
Equation~\eqref{eq:PSR_def} by taking the Fourier-Laplace transform
with respect to $S$ and $R$ respectively, defining thus
\begin{equation}
  \label{eq:FLt_def}
  \widehat{P}(\lambda,q) := \mathbb{E}\left[e^{\iu
        qR -\lambda S}\right]= \int_0^\infty \dint S~\int \dint R~\exp(\iu
  qR-\lambda S)P(S,R)\;\;.
\end{equation}

Indeed, the function $\widehat{P}(\lambda,q)$ may be used to retrieve
several properties of the distribution $P(S,R)$. For example, when
$q=0$, $\widehat{P}(\lambda,q=0)$ becomes the Laplace transform of the
marginal distribution $P(S)$. This can be used to study the behavior
of the firm size distribution in the large $S$ limit, which is
determined by the behavior of $\widehat{P}(\lambda,q=0)$ in the limit
$\lambda \to 0$.

Secondly, Equation~\eqref{eq:FLt_def} implies that
\begin{equation}
  \label{eq:FLt_deriv}
  \left. -\frac{\partial^2
    \widehat{P}(\lambda,q)}{\partial q^2}\right|_{q=0} = \int_0^\infty
  \dint S~ P(S) \exp(-\lambda S) \int \dint R~R^2 P(R|S)\;\;.
\end{equation} Thus, computing the second derivative of the
Fourier-Laplace transform of $P(S,R)$ at $q=0$ yields the conditional
variance of the absolute growth $R$, which will be relevant when
studying the variance of growth rates $R/S$ conditional on firm sizes
$S$.

Thirdly and finally, once $P(S)$ is characterized, plugging
$P(S,R)=P(S)P(R|S)$ into \eqref{eq:FLt_def} leads to
\begin{equation}
  \label{eq:FLt_v2}
  \widehat{P}(\lambda,q) = \int_0^\infty \dint S
  \int \dint R~\exp(\iu qR-\lambda S)P(S)P(R|S)\;\;,
\end{equation}
which can be used to identify the distribution of firm growth rates
conditional on size. Because the properties of the marginal $P(S)$ can
be understood using the technique described above, one is left with
the problem of solving for a marginal distribution $P(R|S)$ to match
the right-hand side of the equation with the left hand side given by
the Fourier-Laplace transform. \medskip

The following lemmas establish two preliminary results providing a
convenient representation of $\widehat{P}(\lambda,q)$ and some
insights on it behavior.
\begin{lemma}
  \label{lemma:1}
  The Fourier-Laplace transform $\widehat{P}(\lambda,q)$ can be represented as
  \begin{equation}
    \label{eq:phat_eval}
    \widehat{P}(\lambda, q)= \sum_{K>1}p(K)\left[1-g(\lambda,q)\right]^K\;\;\;,
  \end{equation}
  where $g(\lambda,q)=\int_0^{\infty}ds\;P(s)\left[1-\exp(-q^2s^2/2-\lambda s)\right]$.
\end{lemma}
\begin{proof}
  Plugging \eqref{eq:PSR_def} into \eqref{eq:FLt_def} gives
  \begin{align*}
    \widehat{P}(\lambda,q) &= \int_0^\infty \dint S~\int \dint R~\exp(\iu qR-\lambda S)
                             \sum_{K>1}p(K)\int_{0}^{^\infty}\prod_{j=1}^K\dint s_j~p(s_j)\delta\left(S-\sum_{j=1}^K s_j\right)\int \prod_{j=1}^K\dint \eta_j~p(\eta_j)\delta\left(R-\sum_{j=1}^K \eta_j s_j\right)\\
                           &=\sum_{K>1}p(K)\left[
                             \int_0^\infty\dint S \int_{0}^{^\infty}\prod_{j=1}^K\dint s_j~p(s_j)\delta\left(S-\sum_{j=1}^K s_j\right)\exp(-\lambda S)\;\int \dint R \int \prod_{j=1}^K\dint \eta_j~p(\eta_j)\delta\left(R-\sum_{j=1}^K \eta_j s_j\right)\exp(\iu qR)
                             \right]\\
                           &=\sum_{K>1}p(K)\left[
                             \int_0^\infty  \prod_{j=1}^K\dint s_j~p(s_j)\int_{0}^{^\infty} \dint S~\delta\left(S-\sum_{j=1}^K s_j\right)\exp(-\lambda S)\;\int\prod_{j=1}^K\dint \eta_j~p(\eta_j)  \int \dint R~\delta\left(R-\sum_{j=1}^K \eta_j s_j\right)\exp(\iu qR)
                             \right]\\
                           &\text{[since $\int \dint x~f(x)\delta(x-a) =f(a)$]}\\
                           &=\sum_{K>1}p(K)\left[
                             \int_0^\infty  \exp\left(-\lambda \sum_{j=1}^K s_j\right) \prod_{j=1}^K\dint s_j~p(s_j)\;\int \exp\left(\iu q\sum_{j=1}^k \eta_js_j\right)\prod_{j=1}^K\dint \eta_j~p(\eta_j)  
                             \right]\\
                           &=\sum_{K>1}p(K)\left[
                             \prod_{j=1}^K\int_0^\infty  \dint s_j~p(s_j)\exp(-\lambda s_j)\;\prod_{j=1}^K\int\dint \eta_j~p(\eta_j)  \exp(\iu q\eta_js_j)
                             \right]\\
                           &\text{[since integrands are independent]}\\
                           &=\sum_{K>1}p(K)\left[
                             \int_0^\infty  \dint s \;\int \dint \eta ~p(s)p(\eta)  \exp(\iu q\eta s -\lambda s)
                             \right]^K\\
                           &\text{[integrating over $\eta$]}\\
                           &=\sum_{K>1}p(K)\left[
                             \int_0^\infty  \dint s~p(s) \exp(-\frac{q^2 s^2}{2} -\lambda s)
                             \right]^K\\
                           &=\sum_{K>1}p(K)\left[1-g(\lambda,q) \right]^K\;\;.
  \end{align*}  
\end{proof}

Lemma~\ref{lemma:1} generates two useful corollaries characterizing
the behaviour of $g(\lambda,q)$ for $q=0$ and in the limit where
$\lambda,q \to 0$ with the product $q\lambda^{-1/\mu}$ constant. The
first exploits the assumption that both the number and the size of
sub-units are Pareto-distributed, with exponents $\alpha$ and $\mu$
respectively and with $1<\alpha<\mu<2$.

\begin{corollary}
  \label{coro:1}
  The Fourier-Laplace transform $\widehat{P}(\lambda,q)$ at $q=0$ reads
  \begin{equation}
    \label{eq:cor1}
    \widehat{P}(\lambda,0) = 1-\lambda\mathbb{E}[s]\mathbb{E}[K] - \lambda^\alpha \mathbb{E}[s]^{\alpha}\Gamma\left[-\alpha+1\right]-\lambda^{\mu}\mathbb{E}[K]\Gamma\left[-\mu+1\right] +\mathcal{O}\left(\lambda^{\max\left(\alpha,\mu\right)}\right)\;\;,
  \end{equation}
  where $\mathbb{E}[s]$ and $\mathbb{E}[K]$ are the expected values of
  the sub-unit size and of the number of sub-units and
  $\Gamma\left[\cdot\right]$ is the gamma function.
\end{corollary}
\begin{proof}
  We start by computing the Laplace transform of the distribution of
  the sub-unit size $P(s)=\dfrac{\mu s_0^\mu}{s^{1+\mu}}$. From
  Lemma~\ref{lem:power_law_laplace} and
  Eq.~\eqref{eq-app:pl_laplace1}, we have directly that
  \begin{align}
    \label{eq-app:pl_laplace}
    g(\lambda,0)=\int_{s_0}^\infty\dint s ~P(s)\exp(-\lambda s) &= 1-\lambda^\mu \Gamma\left[-\mu+1\right]-\lambda \mathbb{E}[s] + \mathcal{O}(\lambda^{\mu}).
  \end{align}
  More generally, the Laplace transform of a power-law distribution
  with a tail $P(s)\sim s^{-1-\mu}$ with $1<\mu<1$, has the form
  $1-A\lambda-B\lambda^\mu+\mathcal{O}(\lambda^\mu)$ with A and B two
  constants that depend on $\mu$. Thus we can obtain the Laplace
  transform of the distribution of the number of sub-units, $p(K)$, by
  changing the roles $\alpha\leftrightarrow \mu$.

  Now, setting $q=0$ in Equation~\eqref{eq:phat_eval} one gets
  \begin{align}
    \widehat{P}(\lambda,0) &\approx \sum_{K>1}p(K)\left[1-g(\lambda,0) \right]^K \nonumber \\
    &\approx \sum_{K>1}p(K)\left[1-\lambda^\mu \Gamma\left[-\mu+1\right]-\lambda \mathbb{E}[s] \right]^K \nonumber\\
    &\text{[approximating $(1-x)^K$ with $\exp(-Kx)$]} \nonumber\\
    &\approx \sum_{K>1}p(K) \exp(-K\lambda^\mu \Gamma\left[-\mu+1\right]-K\lambda \mathbb{E}[s]) \nonumber\\
    &\text{[setting $\upsilon=\lambda^\mu \Gamma\left[-\mu+1\right]+\lambda \mathbb{E}[s]$]} \nonumber\\
    &\approx \int_0^\infty\dint K~ P(K) \exp(-\upsilon K) \nonumber\\
    &\text{[using Equation~\eqref{eq-app:pl_laplace} and Assumption~\ref{hp:su_size}]}\nonumber\\
    &\approx 1-\upsilon^\alpha \Gamma\left[-\alpha+1\right]-\upsilon \mathbb{E}[K]\nonumber\\
    &\approx 1-(\lambda^\mu \Gamma\left[-\mu+1\right]+\lambda \mathbb{E}[s])^\alpha \Gamma\left[-\alpha+1\right]-(\lambda^\mu \Gamma\left[-\mu+1\right]+\lambda \mathbb{E}[s]) \mathbb{E}[K]\nonumber\\
    &\text{[since, by Assumption~\ref{hp:su_n_gabaix}, $1<\alpha<\mu<2$]} \nonumber\\
    &\approx 1-\lambda^\alpha \mathbb{E}[s]^{\alpha}\Gamma\left[-\alpha+1\right] -\lambda^\mu\mathbb{E}[K]\Gamma\left[-\mu+1\right] -\lambda\mathbb{E}[s]\mathbb{E}[K] - +\mathcal{O}(\lambda^{\alpha})\;\; \nonumber.
  \end{align}
\end{proof}

Note that this result intuitively implies that
$\mathbb{E}[S]=\mathbb{E}[s]\mathbb{E}[K]$, which means that the
average size of a firm is given by the product of the average number
of subunits in a firm by the average size of a subunit. This can be
obtained directly by noting that
$\mathbb{E}[S]=\frac{\partial \widehat{P}(\lambda ,0)}{\partial
  \lambda}\big\vert_{\lambda=0}$.

The second corollary describes the behaviour of
$\widehat{P}(\lambda,q)$ in the limit where $\lambda,q \to 0$, but
keeping $q\lambda^{-1/\mu}$ constant. This corollary will later allow
us to compute the distribution $P(R|S)$ in the scaling region where
$R$ is proportional to $S^{\frac{1}{\mu}}$.
\begin{corollary}
  \label{coro:2}
  The Fourier-Laplace transform $\widehat{P}(\lambda,q)$ computed in
  the limit $\lambda,q \to 0$ where $q\lambda^{-1/\mu}=\theta$ reads
  \begin{equation}
    \label{eq:cor2}
    \widehat{P}(\lambda,q=\theta\lambda^{1/\mu}) = 1-\lambda\mathbb{E}[K]\left(\mathbb{E}[s] + \mu s_0^{\mu}\theta^{\mu}I(\mu)\right)+\mathcal{O}(\lambda)\;\;,
  \end{equation}
  where $I(\mu)= \int_0^\infty \dint t~t^{-(1+\mu)}(1-e^{-t^2/2})=-2^{-(1+\mu/2)}\Gamma(-\frac{\mu}{2})$.
\end{corollary}
\begin{proof}
  As for the previous lemma we start by computing
  \begin{align}
    \label{eq:cor2-b}
    g(\lambda,q=\theta \lambda^{1/\mu}) &= \int_{0}^{\infty}\dint s~P(s)\left(1-\exp\left(-\frac{\theta^2 \lambda^{\frac{2}{\mu}}s^2}{2} - \lambda s\right)\right)  \nonumber \\
                                        &\text{[using the linear approximation of $-e^{-x}$]}  \nonumber \\
                                        &= \int_{0}^{\infty}\dint s~P(s)\left(1-\exp\left(-\frac{\theta^2 \lambda^{\frac{2}{\mu}}s^2}{2}\right) +\lambda s \right)+\mathcal{O}(\lambda)  \nonumber \\
                                        &= \lambda \mathbb{E}[s] + \int_{0}^{\infty} \dint s~\frac{\mu s_0^\mu}{s^{1+\mu}} \left(1-\exp\left(-\frac{\theta^2 \lambda^{\frac{2}{\mu}}s^2}{2}\right)\right)+\mathcal{O}(\lambda)  \nonumber \\
                                        &= \lambda \left(\mathbb{E}[s] + \mu s_0^\mu \theta^\mu \int_{0}^{\infty} \frac{\dint t}{t^{1+\mu}}\left(1-e^{-\frac{t^2}{2}}\right) \right)+\mathcal{O}(\lambda) \nonumber \\
                                        &= \lambda \left(\mathbb{E}[s]+\mu(s_0\theta)^{\mu}I(\mu)\right) +\mathcal{O}(\lambda)\;\;.
  \end{align}
  Plugging \eqref{eq:cor2-b} in Equation~\eqref{eq:phat_eval} one gets
  \begin{align}
    \label{eq:target}
    \widehat{P}(\lambda,q=\theta\lambda^{1/\mu}) & = \sum_{K>1} p(K) [1-\lambda \left(\mathbb{E}[s]+\mu(s_0\theta)^{\mu}I(\mu)\right)+\mathcal{O}(\lambda)]^K \nonumber \\
                                                 &\text{[approximating $(1-x)^K$ with $1-Kx$ and evaluating the sum with its integral representation]} \nonumber\\
                                                 & = 1- \lambda\mathbb{E}[K]\left(\mathbb{E}[s]+\mu(s_0\theta)^{\mu}I(\mu)\right)+\mathcal{O}(\lambda)\;\;.
  \end{align}
  
\end{proof}

Finally, the second Lemma states that the second derivative of the
Fourier-Laplace transform
$\frac{\partial^2 \hat{P}(\lambda,q)}{\partial q^2}$, evaluated at
$q=0$, is proportional to $\lambda^{\mu-2}$. More generally,
$\frac{\partial^{2\ell} \hat{P}(\lambda,q)}{\partial q^{2\ell}}\sim
\lambda^{\mu-2\ell}$ at $q=0$ and as $\lambda\to 0$.
\begin{lemma}
  \label{lemma:2}
  When $1<\mu<2$,
  $-\dfrac{\partial^2 \hat{P}(\lambda,q)}{\partial q^2}$ evaluated at
  $q=0$ behaves as $\lambda^{\mu-2}$, that is
  \begin{equation}
    \label{eq:deriv2}
    \left.\dfrac{\partial^2 \hat{P}(\lambda,q)}{\partial q^2} \right|_{q=0} \sim \lambda^{\mu-2}.
  \end{equation}
\end{lemma}
\begin{proof}
  Plug Equation~\eqref{eq:phat_eval} in $-\dfrac{\partial^2
    \hat{P}(\lambda,q)}{\partial q^2}$ to obtain
  \begin{align*}
    \left.-\dfrac{\partial^2 \hat{P}(\lambda,q)}{\partial q^2}\right|_{q=0}&=
    2  \left.\dfrac{\partial}{\partial (q^2)} \sum_{K>1} p(K)[1-g(\lambda,q)]^K\right|_{q=0}\\
    &\text{[using the exponential approximation]} \nonumber\\
    &=-2 \left.\dfrac{\partial}{\partial (q^2)} \sum_{K>1} p(K) \exp(-g(\lambda,q)K)\right|_{q=0}\\
    &=2 \left.\dfrac{\partial g(\lambda,q)}{\partial (q^2)}\right|_{q=0} \sum_{K>1} Kp(K) \exp(-Kg(\lambda,0))\;\;.\\
  \end{align*}
  Computing the derivative at $q=0$ gives
  \begin{align*}
  &\left.2\dfrac{\partial g(\lambda,q)}{\partial (q^2)}\right|_{q=0} 
   =\int_0^{\infty}\dint s\;P(s) s^2 e^{-\lambda s}  
    \approx \int_0^{\infty} \dint s~ s_0^\mu s^{1-\mu} e^{-\lambda s} \underset{\lambda \to 0}{\approx} s_0^\mu \frac{\Gamma\left[2-\mu\right]}{\lambda^{2-\mu}}\;\;.
  \end{align*}
  Thus we find, for $\lambda \to 0$ 
  \begin{align*}
    \left.-\dfrac{\partial^2 \hat{P}(\lambda,q)}{\partial q^2}\right|_{q=0}&\sim s_0^\mu \frac{\Gamma\left[2-\mu\right]}{\lambda^{2-\mu}} \sum_{K>1} Kp(K) \exp(-Kg(\lambda,0)) \\
    &\sim s_0^\mu \frac{\Gamma\left[2-\mu\right]}{\lambda^{2-\mu}} \sum_{K>1} Kp(K) e^{-K} \sim \lambda^{\mu-2}\;\;,
  \end{align*}
  where we've approximated $e^{-Kg(\lambda,0)}\approx e^{-K}$ and used
  that $\sum_{K>1}Kp(K)e^{-K}$ is a finite constant.
\end{proof}

\subsection{Proof of Proposition \ref{prop:fsd}.}
\label{app:fsd}
\setcounter{proposition}{2}

\begin{proposition}
  The firm size distribution behaves, for large $S$, as the sum of the
  two Pareto distributions
  \begin{equation} 
    P(S_i) \sim \frac{C_\alpha}{S_i^{1+\alpha}} + \frac{C_\mu}{S_i^{1+\mu}},
  \end{equation}
  where $C_\alpha = \left(\frac{\mu}{\mu-1}\right)^\alpha$ and
  $C_\mu=\frac{\alpha}{\alpha-1}$.
\end{proposition}
\begin{proof}
  This is a direct consequence of Corollary~\ref{coro:1}. Since
  $\widehat{P}(\lambda,0)$ is the Laplace transform of the firm size
  distribution $P(S)$, for sufficiently large $S$ its behavior is
  described by Equation~\eqref{eq:cor1}. Gathering the terms
  proportional to $\lambda^{\alpha}\Gamma\left[-\alpha+1\right]$ and
  $\lambda^\mu \Gamma\left[-\mu+1\right]$ allows us to find the
  asymptotic behaviour described above and the relative contributions
  of each component by identifying the Laplace transforms of power-law
  tails. Since the inverse of the Laplace transform of a probability
  distribution is unique, Proposition~\ref{prop:fsd} follows. Noting
  that the last line implies that the firm size distribution is
  asymptotically proportional to

  \begin{equation}
    \label{eq:firm_size_etc}
    P(s)\sim \mathbb{E}[s]^\alpha S^{-1-\alpha} + \mathbb{E}[K]S^{-1-\mu},
  \end{equation}
  by getting the terms proportional to
  $\lambda^{\alpha}\Gamma\left[-\alpha+1\right]$ and
  $\lambda^\mu \Gamma\left[-\mu+1\right]$ in the expansion above, we
  may also write this as
  \begin{equation}
    \label{eq:firm_size_contributions}
    P(s) \sim S^{-1-\alpha}\left(1+C S^{\alpha-\mu}\right),
  \end{equation} 
  where $C$ is a constant. 
\end{proof}

\subsection{Proof of Proposition~\ref{prop:vol_dist}}
\label{app:vol_dist}

\begin{proposition}
  The distribution of growth rate volatilities conditional on size $S$ is given, for large $S$, by:
  \begin{equation}
    P(\sigma \vert S) \approx
    \frac{1-\pi(S)}{\overline{\sigma}(S)}
    G\left(\frac{\sigma}{\overline{\sigma}(S)}\right)\left(1-\frac{\sigma}{\sigma_0}\right)^{\mu-1}
    +\pi(S)H(\sigma),
  \end{equation}
  where $G(\cdot)$ is defined in Eq. \eqref{eq:Herf2},
  $\overline{\sigma} \sim S^{-\beta}$ is given by
  Eq. \eqref{eq:sigma_scal} with $\beta=(\mu-1)/\mu$,
  $\pi(S)\sim S^{\alpha-\mu}$ and finally $H(\cdot)$ is a contribution
  peaked at $\sigma \approx \sigma_0$. In particular,
  $G(x)\sim x^{-1-\mu}$.
\end{proposition}

\begin{proof}
  Following Proposition~\ref{prop:fsd}, firms of size $S$ are made up
  of two distinct fractions: those for which
  $S\approx K \mathbb{E}[s]$, which corresponds to the $S^{-1-\alpha}$
  tail of the distribution with the exponent $\alpha$ indicating that
  firms in that tail are large because they have a large number of
  sub-units, and those firms for which $K=\mathcal{O}(1)$ but are
  still large because one single sub-unit is large, i.e.
  $S\approx s_{\max}$, thereby contributing with a tail with weight
  $S^{-1-\mu}$. The latter firms have $\mathcal{H}\approx 1$.

  Because of the relative contributions of the tails, the fraction of
  poorly diversified firms is $\pi(S)\sim S^{\alpha-\mu}$, as
  highlighted above. We may therefore write that

  \begin{equation}
    P(\sigma |S) = (1-\pi(S)) P(\sigma |S, \text{many subunits}) + \pi(S) P(\sigma |S, \text{few subunits}).
  \end{equation}
\end{proof}

Because poorly diversified firms have $K$ of order $1$ and are
concentrated in a single sub-unit, their Herfindahl is of order $1$
and so $P(\sigma|S, \text{poorly diversified}) = H(\sigma)$ is a
distribution peaked at $\sigma_0$ for all $S$. Regarding the
contribution of firms with $S\propto K$, their Herfindahl is
distributed as described in Eq.~\eqref{eq:Herf2}. It follows that

\begin{equation}
  P(\sigma|S, \text{well diversified}) = \frac{1}{\overline{\sigma}(S)}G\left(\frac{\sigma}{\overline{\sigma}(S)}\right)\left(1-\frac{\sigma}{\sigma_0}\right)^{\mu-1},\quad G(x) = 2x F(x^2),
 \end{equation} 
 where $F(\cdot)$ is the function defined in
 Eq.~\eqref{eq:Herf}. Since $F(x)\sim x^{-1-\mu/2}$ it follows that
 $G(x)\sim x^{-1-\mu}$ for large $x$. Proposition~\eqref{prop:vol_dist}
 follows.

\subsection{Proof of Proposition~\ref{prop:vol_moments}.}
\label{app:vol_moments}

\begin{proposition}
    For $1 \leq \alpha < \mu$, the integer moments of the growth rate
  volatilities conditional to size $S$ are asymptotically given, for
  large $S$, by:
  \begin{equation}
    \label{eq-app:vol_moments}
    \mathbb{E}[\sigma^q|S] = C_1 S^{1-\mu} + C_2 S^{q\frac{1-\mu}{\mu}} + C_3 S^{\alpha-\mu} +  \mathcal{O}\left(S^{\min (\alpha-\mu, 1-\mu, q\frac{1-\mu}{\mu})}\right)\;\;,
  \end{equation}
  where $C_1$, $C_1$ and $C_3$ are numerical constants.\footnote{Note
    that for $\alpha = -1$, corresponding to Gabaix' model,
    Eq. \eqref{eq-app:vol_moments} precisely recovers
    Eq. \eqref{eq:gen_moments_body}, as it should.}
\end{proposition}
\begin{proof}
  As before, we may write this as
  \begin{equation}
    \mathbb{E}\left[\sigma^q | S\right] = (1-\pi(s)) \int \dint \sigma ~ \sigma^q P(\sigma | S, \text{many subunits}) +  \pi(S)\int \dint \sigma~\sigma^q P(\sigma| S, \text{few subunits}), 
  \end{equation}
  and where the second contribution corresponds to firms with
  $\mathcal{H}\approx 1$. The second contribution will therefore
  always give
  \begin{equation}
    \pi(S)\int \dint \sigma~\sigma^q P(\sigma| S, \text{poorly diversified}) \approx \pi(S) \sim S^{\alpha-\mu}.
  \end{equation}
  The case of firms with many subunits is slightly different. Noting
  that $\sigma \propto \mathcal{H}^{1/2}$, we use the result given in
  Proposition~\ref{prop:herfindahl_gabaix}, substitute $q\to q/2$ and
  use $K\propto S$ to obtain
  \begin{equation}
    \label{eq:many_subunits_q}
    \begin{split}
      (1-\pi(s)) \int \dint \sigma ~ \sigma^q P(\sigma | S, \text{many subunits}) \approx& \int \dint \sigma ~ \sigma^q P(\sigma | S, \text{many subunits})\\
      \approx & C_1 S^{1-\mu} + C_2 S^{q\beta } +\mathcal{O}\left(S^{\min\left(1-\mu, q\beta\right)}\right).
    \end{split}
  \end{equation}
  Finally, we have that 
  \begin{equation}
    \mathbb{E}[\sigma^q|S] = C_1 S^{1-\mu} + C_2 S^{q\beta } + C_3 S^{\alpha-\mu} + \mathcal{O}\left(S^{\min (\alpha-\mu, 1-\mu, q\beta)}\right).
  \end{equation}
  The fact that the dominant term is $\propto S^{\alpha-\mu}$, as
  predicted by the equation above, can also be proven from
  Equation~\eqref{eq:FLt_deriv}. Consider the term
  $\int \dint R~ R^2 P(R|S)$, which represents the conditional
  variance of a firm's absolute growth, $\mathbb{V}[R|S]$. We will
  surmise that $\mathbb{V}[R|S]\propto S^b$ for some $b$ and find the
  value of $b$ such that this integral scales as $\lambda^{\mu-2}$ as
  stated in Lemma~\ref{lemma:2}. Doing this substitution in
  Equation~\eqref{eq:FLt_deriv} leads, when $b > \alpha$, to
  \begin{align*}
    -\left.\dfrac{\partial^2 \hat{P}(\lambda,q)}{\partial q^2} \right|_{q=0} &\propto
     \int_0^\infty \dint S~ S^{-(\alpha+1)}  ~S^b ~\exp(-\lambda S) = \Gamma\left[b-\alpha\right]\lambda^{\alpha-b}\;\;,
  \end{align*}
  which combined with Equation~\eqref{eq:deriv2} in Lemma~\ref{lemma:2} implies that
  \begin{align*}
    \lambda^{\mu-2} \propto \lambda^{\alpha-b} \qquad \Rightarrow \qquad b=\alpha-\mu+2\;\;,
  \end{align*}
  which is larger than $\alpha$ when $\mu < 2$.
  This proves that $\mathbb{V}[R|S] \propto S^{\alpha-\mu+2}$ and gives directly
  \begin{align*}
      \mathbb{E}[\sigma^2 |S] = \frac{\mathbb{E}\left[R^2|S\right]}{S^2} \propto S^{\alpha-\mu}\;\;,
  \end{align*}
as wanted.
\end{proof}

\subsection{Proof of Proposition~\ref{prop:gmm_levy}.}
\label{app:gmm_levy}
\begin{proposition}
  The distribution of growth rates conditional on size $S$ and on
  growth volatility $\sigma$ is given, for large $S$, by:
  \begin{equation}
    P_S(g \vert \sigma,S) \approx \left(1 - \pi(S)\right)\;\mathcal{N}(0,\sigma^2) +  \pi(S)\;Q_{\eta} \;\;,
  \end{equation}
  where $\mathcal{N}(0,\sigma^2)$ is a Gaussian distribution with variance
  $\sigma^2=\sigma^2_0\mathcal{H}$ and $Q_{\eta}$ a non universal
  distribution that depends on the distribution of the sub-unit growth
  shocks $\eta$. The weight $\pi(S)$ represents the probability of
  observing a large firm with only few sub-units vanishing when size
  grows larger as $S^{\alpha-\mu}$.

  Neglecting large firms with a small number of sub-units and
  integrating Eq.~\eqref{eq:WB_GMM} over the first term of the
  distribution $P(\sigma|S)$ in Eq.~\eqref{eq:cond_vol} gives
  \begin{equation}
    P\left(g|S\right) \sim
    S^{\beta}L_{\mu, 0}(gS^{\beta}),\;\;\; \text{when} \;\; g \ll 1,
  \end{equation}
  where $L_{\mu,0}(\cdot)$ is the symmetric L\'evy alpha-stable
  probability density with stability parameter
  $1 < \mu < 2$. Because of the cut-off in the distribution of $\sigma$, this distribution also has a cut-off, with $P(g|S)=0$ for $g\gtrapprox S$. The complete distribution $P(g)$ is obtained by integrating over $P(s)$, and behaves asymptotically as $P(g)\sim \vert g\vert^{-1-\mu}$.
\end{proposition}
\begin{proof}
  Eq.~\eqref{eq:WB_GMM} states simply that, conditional on the
  size $S$, a fraction $1-\pi(S)$ of firms will be such that
  $S\propto K$ and will be made up of $K\gg 1$ sub-units. These firms
  will have growth rates that are the sum of $K\gg1$ sub-unit growth
  rates, resulting in a Gaussian density of variance $\sigma^2$. The
  second contribution corresponds to firms with $K=\mathcal{O}(1)$,
  and therefore the central limit theorem does not apply and the
  distribution is closer to that of the individual sub-units' growth
  rates $\eta_{ijt}$.

Obtaining the distribution of $P(g|S)$ is in principle possible by integrating directly over the volatility $\sigma$ with its probability distribution $P(\sigma)$ which has been discussed above. We proceed instead by studying the distribution $P(S,R)$ introduced above, finding first the distribution of $R|S$, the absolute size change conditional on size, which will then allow to find the distribution $P(g|S)$.

Writing $P(S,R)=P(S)P(R|S)$, using the definitions from~\eqref{eq:PSR_def}, we have that
  \begin{align}
    \label{eq:arrow}
    \widehat{P}(\lambda,q=\theta\lambda^{1/\mu})&=\int_0^{+\infty}\dint S \int\dint R~P(S)P(R|S)\exp(i\theta\lambda^{1/\mu}R-\lambda S)\;\;.
  \end{align}
  Next, we surmise that when
  $P(R|S)=S^{-1/\mu}L_{0,\mu}(RS^{-1/\mu})$, i.e. when $R|S$ is distributed according to a symmetric Lévy distribution with parameter $\mu$,
  Equation~\eqref{eq:arrow} has the same form as that found in
  Equation~\eqref{eq:cor2} of Corollary~\ref{coro:2}.

  Indeed, setting $\upsilon=RS^{-1/\mu}$ in \eqref{eq:arrow} we obtain
  \begin{align*}
    \widehat{P}(\lambda,q=\theta\lambda^{1/\mu})&=\int_0^{+\infty}\dint
    S~P(S) \int\dint
    \upsilon~L_{0,\mu}(\upsilon)\exp(i\theta\lambda^{1/\mu}S^{1/\mu}\upsilon-\lambda
    S) \\ &=\int_0^{+\infty}\dint S~P(S) \exp(-\lambda S)\int\dint
    \upsilon~L_{0,\mu}(\upsilon)\exp(i\theta\lambda^{1/\mu}S^{1/\mu}\upsilon)
  \end{align*}
  where one notes that the integral in $\upsilon$ is the
  characteristic function or Fourier transform of a symmetric
  L\'evy alpha-Stable distribution of parameter $\mu$, $L_{0,\mu}$,
  evaluated at $\theta \lambda^{1/\mu}S^{1/\mu}$.  This characteristic
  function behaves as $\int \mathrm{d}\upsilon~e^{\iu t
    \upsilon}\propto \exp(-At^{1/\mu})$, where $A$ is a
  constant. Applying this to the equation above leads to:
  \begin{align*}
    \widehat{P}(\lambda,q=\theta\lambda^{1/\mu})&=\int_0^{+\infty}\dint S~P(S) \exp(-\lambda S)~\exp(-A\theta^{\mu}\lambda S).
  \end{align*}
  Then
   \begin{align*}
    \widehat{P}(\lambda,q=\theta\lambda^{1/\mu})&=\int_0^{+\infty}\dint S~P(S) \int\dint \upsilon~L_{0,\mu}(\upsilon)\exp(i\theta\lambda^{1/\mu}S^{1/\mu}\upsilon-\lambda S) \\
    &=\int_0^{+\infty}\dint S~P(S) \exp(-\lambda S)\int\dint \upsilon~L_{0,\mu}(\upsilon)\exp(i\theta\lambda^{1/\mu}S^{1/\mu}\upsilon) \\
    &\text{[the integral in $\upsilon$ is the Fourier transform of $L_{0,\mu}$ which has an exponential form]} \nonumber\\
    &=\int_0^{+\infty}\dint S~P(S) \exp(-\lambda S)~\exp(-A\theta^{\mu}\lambda S) \\
    &\text{[where A is a constant]} \nonumber\\
    &=\int_0^{+\infty}\dint S~P(S) \exp(-\lambda S (1+A\theta^{\mu})) \\
    &=1-\int_0^{+\infty}\dint S~P(S) \left( 1-\exp(-\lambda S (1+A\theta^{\mu}))\right) \\
    &\text{[using twice the linear approximation of $-e^{-x} \sim x-1$]}\\
    &=1-\lambda\mathbb{E}\left[S\right] - A\theta^{\mu}\lambda \mathbb{E}\left[S\right] + \mathcal{O}(\lambda) \\
    &=1-\lambda\mathbb{E}\left[S\right](1+A\theta^{\mu}) + \mathcal{O}(\lambda) \;\;,
  \end{align*}
  which indeed matches Equation~\eqref{eq:cor2}.

  Now, this proves that $RS^{-1/\mu}$ is distributed according to a Lévy distribution with parameter $\mu$, and therefore

\begin{equation}
P(R|S) \sim S^{-1/\mu}L_{\mu,0}(RS^{-1/\mu}).
\end{equation}
since $R=gS$ and since $\beta=1-1/\mu$, this proves the first part of the proposition. The cut-off value at $g\sim 1$ results from the fact that the Herfindahl has a cut-off at $\mathcal{H}=1$.

For the second part, we first write that $P(g)=\int \dint S~P(S)P(g|S)$, and then consider the cumulative density function, $P_>(g):=\int_{g}^\infty \dint g_1~P(g_1)$. This yields

\begin{equation}\label{eq:g_proof}
\begin{split}
P_>(g)=& \int_g^\infty \dint g_1 \int \dint S~ P(S) S^{\beta}L_{0,\mu}(g_1 S^\beta)\\
=& \int \dint S~ P(S) S^{\beta} \int_{g}^\infty \dint g_1 ~L_{0,\mu}(g_1 S^{\beta})\\
&\text{(using that asymptotically, $L_{0,\mu}(x)\sim x^{-1-\mu}$)}\\
\sim& \int \dint S~ P(S) S^{\beta} S^{-\beta(1+\mu)}\int_{g}^\infty \dint g_1 g_1^{-1-\mu}\\
\sim& g^{-\mu} \int \dint S~P(S)S^{-\beta\mu}\sim g^{-\mu},
\end{split}
\end{equation}
where we have used the fact that $\beta \mu = \mu - 1 < 1$ so that the integral over $S$ is convergent for both small and large $S$. This shows that $P(g)$ is a mixture of Lévy alpha-stable distributions with a power-law tail with exponent $\mu$.

\end{proof}

\subsection{Proof of Proposition~\ref{pred:agg}}
\label{app:agg}

\begin{proposition}
  Proposition~\ref{prop:fsd}, \ref{prop:vol_dist},
  \ref{prop:vol_moments} and \ref{prop:gmm_levy} are robust against
  the additive aggregation of firms into (possibly fictitious)
  supra-firms.
\end{proposition}
\begin{proof} Consider two firms, $i$ and $j$, with $K_i$ and $K_j$
  sub-units respectively. Since both are Pareto distributed with
  exponent $\alpha$, the additive aggregation of these two firms
  produces a new entity with $K_i+K_j$ sub-units, which is also
  asymptotically Pareto with the same exponent $\alpha$. The size
  distribution of the sub-units clearly remains Pareto distributed as
  well. As a consequence, starting with $N$ firms which are then
  merged into $N/n$ ``super-firms'' leads to the same statistical
  properties for size, growth and the size/growth-volatility
  relationship discussed above.
\end{proof}

\newpage

\section{Descriptive statistics and further investigations}
\label{app:f-inv}

\subsection{Main sample}
\label{app:main}

\noindent In this Appendix we report descriptive statistics on size
($\tilde{S}$), normalized size ($S$) and on their $(\log)$ growth
rates for the main sample used in the paper. We also discuss the
existence of the hump $H(\sigma)$ predicted by
Proposition~\ref{prop:vol_dist} and explore the consequences of using
the standard deviation to proxy for volatility.

\medskip

\noindent {\bf Descriptive statistics}. In Table~\ref{tab:sample-desc}
statistics are reported at the firm-year level. This sample contains
more thna 1.2 millions of firm-year observations, with an extreme
range of variation both for sizes and growth rates and for the three
variants (nominal, deflated and normalized). As expected different
variants of growth rates do not differ much and, less expected, their
range of variation is well centered around zero.

\begin{table}[H]
  \centering 
  \caption{Descriptive statistics on size and growth rates (main sample)} 
  \label{tab:sample-desc}
  \scalebox{0.6}{    
    \begin{threeparttable}
\begin{tabular}{@{\extracolsep{5pt}}lccccc} 
\\[-1.8ex]\hline 
\hline \\[-1.8ex] 
 & \multicolumn{1}{c}{N} & \multicolumn{1}{c}{Mean} & \multicolumn{1}{c}{St. Dev.} & \multicolumn{1}{c}{Min} & \multicolumn{1}{c}{Max} \\ 
\hline \\[-1.8ex] 
size $\tilde{S}$ & 1,228,260 & 348,088,528.000 & 2,050,671,382.000 & 1,000.000 & 152,079,000,000.000 \\ 
size $\tilde{S}$ deflated & 1,228,260 & 481,340,797.000 & 2,427,441,597.000 & 940.955 & 147,203,757,866.000 \\ 
size $S$ & 1,228,260 & 1.001 & 4.734 & 0.00000 & 279.152 \\ 
growth rate $\tilde{g}$ & 1,121,115 & 0.094 & 0.591 & $-$12.324 & 12.435 \\ 
growth rate $\tilde{g}$ deflated & 1,121,115 & 0.066 & 0.590 & $-$12.386 & 12.403 \\ 
growth rate $g$ & 1,121,115 & 0.045 & 0.593 & $-$12.070 & 12.318 \\ 
\hline \\[-1.8ex] 
\end{tabular} 
      \begin{tablenotes}
        \footnotesize
      \item {\it Notes}: Growth rates are computed as logarithmic
        differences. Data source: Compustat. \par
      \end{tablenotes}
    \end{threeparttable}
  }
\end{table}

\noindent Table~\ref{tab-app:sample-desc-f} reports statistics at the
firm level. In this sample there are $24,233$ different firms. Each
firm has on average $46$ growth rates whose volatility, proxied by the
corresponding Mean Absolute Deviation (mad) adjusted by
$\sqrt{\pi/2}$, ranges from $0.002$ to $5.9$ with an average of
$0.48$.

\begin{table}[H]
  \centering 
  \caption{Descriptive statistics on size and growth volatility (main sample)} 
  \label{tab-app:sample-desc-f}
  \scalebox{0.6}{
    \begin{threeparttable}  
\begin{tabular}{@{\extracolsep{5pt}}lccccc} 
\\[-1.8ex]\hline 
\hline \\[-1.8ex] 
 & \multicolumn{1}{c}{N} & \multicolumn{1}{c}{Mean} & \multicolumn{1}{c}{St. Dev.} & \multicolumn{1}{c}{Min} & \multicolumn{1}{c}{Max} \\ 
\hline \\[-1.8ex] 
size $S$ & 24,233 & 0.484 & 2.466 & 0.00000 & 109.177 \\ 
num growth rates & 24,233 & 46.264 & 44.972 & 2 & 236 \\ 
growth volatility (mad) & 24,233 & 0.479 & 0.490 & 0.002 & 5.879 \\ 
growth volatility (sd) & 24,233 & 0.532 & 0.525 & 0.002 & 5.753 \\ 
\hline \\[-1.8ex] 
\end{tabular} 
      \begin{tablenotes}
        \footnotesize
      \item {\it Notes}: ``mad'' and ``sd'' represent the mean absolute
        deviation adjusted by the factor $\sqrt{\pi/2}$ and the
        standard deviation. Data source: Compustat. \par
      \end{tablenotes}
    \end{threeparttable}
  }
\end{table}

\noindent Table~\ref{tab-app:gr-binned} reports statistics for
normalized growth rates $\tilde{g}$ grouped in the $25$ size bins
built on the average normalized size $\bar{S}_i$. By construction each
bin contains the same number of firms, but a significantly different
number of observations increasing with size. Looking at the range of
variation we observe few extreme realizations which we have checked do
not drive our main results.

\begin{table}[H]
  \centering 
  \caption{Rescaled growth rates by size bin (main sample)} 
  \label{tab-app:gr-binned}
  \scalebox{0.6}{
    \begin{threeparttable}
\begin{tabular}{@{\extracolsep{5pt}} l cccccccc} 
\\[-1.8ex]\hline 
\hline \\[-1.8ex] 
bin & n firms & n obs. & n distinct obs. & mean & median & sd & min & max \\ 
\hline \\[-1.8ex] 
$1$ & $956$ & $15,429$ & $15,162$ & $0.035$ & $$-$0.023$ & $5.395$ & $$-$122.467$ & $612.003$ \\ 
$2$ & $956$ & $22,036$ & $21,884$ & $$-$0.051$ & $$-$0.034$ & $8.903$ & $$-$1,298.563$ & $88.184$ \\ 
$3$ & $956$ & $25,541$ & $25,474$ & $0.0004$ & $$-$0.021$ & $1.707$ & $$-$122.556$ & $39.474$ \\ 
$4$ & $956$ & $29,843$ & $29,797$ & $$-$0.002$ & $$-$0.033$ & $1.482$ & $$-$46.430$ & $46.910$ \\
\end{tabular} 
      \begin{tablenotes}
        \footnotesize
      \item {\it Notes}: Descriptive statistics of heterogeneously
        rescaled growth rates $\hat{g}$ binned according to their
        average size $\bar{S}$. Data source: Compustat. \par
      \end{tablenotes}
    \end{threeparttable}
  }
\end{table}

\noindent{\bf Missing hump in
  $H(\sigma)$}. Figure~\ref{fig-app:pred1_hump} reports, on a double
log scale, the distributions of the growth rate volatility for various
size bins. We see no sign of an hump for large volatilities and we
confirm visual inspection with the non-parametric test of bi-modality
presented in \cite{ameijeiras_et_al_2019} which rejects the existence
of a second mode at any reasonable level of statistical significance
for all bins. Together with bin 1, 5, 10, 15, 20 and 25 we report bin
8 and 9 associated with the lower p-value for the null hypothesis that
the true number of modes is 1.

\begin{figure}[H]
  \caption{Growth volatility distribution by size bin (main sample)}
  \label{fig-app:pred1_hump}
  \begin{center}    
    \scalebox{1}{
      \input{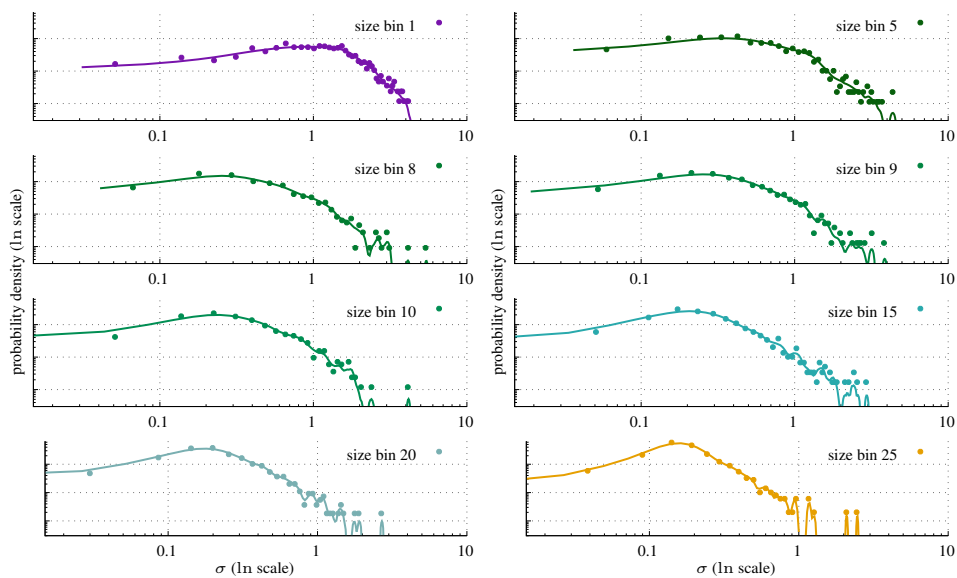}
      }
  \end{center}
  {\scriptsize {\it Notes}. Panels report for different bins defined
    in term of normalized size, the distribution of the growth rate
    volatility on a double log scale. Points represent the histogram
    while solid lines the kernel density estimates. In all the
    figures volatility is computed as the mean absolute deviation
    multiplied by the factor $\sqrt{\pi/2}$. Data source:
    Compustat. \par}
\end{figure}

\begin{table}[H] \centering 
  \caption{ACR non parametric test of bi-modality for the growth volatility distribution} 
  \label{tab-app:acr}
  \scalebox{0.6}{
\begin{tabular}{@{\extracolsep{5pt}} cccc|cccc|cccc} 
\\[-1.8ex]\hline 
\hline \\[-1.8ex] 
bin & ACR statistic & p.value & n & bin & ACR statistic & p.value & n & bin & ACR statistic & p.value & n \\ 
\hline \\[-1.8ex] 
1 & 0.013 & 0.944 & 970 & 11 & 0.013 & 0.784 & 969 & 21 & 0.016 & 0.016 & 0.016 \\ 
2 & 0.015 & 0.66 & 970 & 12 & 0.016 & 0.34 & 969 & 22 & 0.015 & 0.015 & 0.015 \\ 
3 & 0.014 & 0.684 & 970 & 13 & 0.013 & 0.846 & 969 & 23 & 0.017 & 0.017 & 0.017 \\ 
4 & 0.017 & 0.247 & 970 & 14 & 0.014 & 0.689 & 969 & 24 & 0.016 & 0.016 & 0.016 \\ 
5 & 0.013 & 0.786 & 970 & 15 & 0.017 & 0.2 & 969 & 25 & 0.011 & 0.011 & 0.011 \\ 
6 & 0.014 & 0.7 & 970 & 16 & 0.013 & 0.799 & 969 &  &  &  &  \\ 
7 & 0.013 & 0.737 & 970 & 17 & 0.012 & 0.977 & 969 &  &  &  &  \\ 
8 & 0.02 & 0.082 & 970 & 18 & 0.014 & 0.635 & 969 &  &  &  &  \\ 
9 & 0.02 & 0.069 & 969 & 19 & 0.017 & 0.253 & 969 &  &  &  &  \\ 
10 & 0.012 & 0.961 & 969 & 20 & 0.018 & 0.222 & 969 &  &  &  &  \\ 
\hline \\[-1.8ex] 
\end{tabular} 
  }
\end{table}

\noindent {\bf Growth volatility estimated using the standard
  deviation}. Figure~\ref{fig-app:pred1_sd} and
Figure~\ref{fig-app:scaling_sd} report the key empirical results on
growth volatility when we replace the mean absolute deviation (mad)
with the standard deviation (sd). The shape of its distribution, its
scaling with size and the quantitative fir of the Modified Inverse
Gamma distribution (MIG) remain qualitatively the same.

\begin{figure}[H]
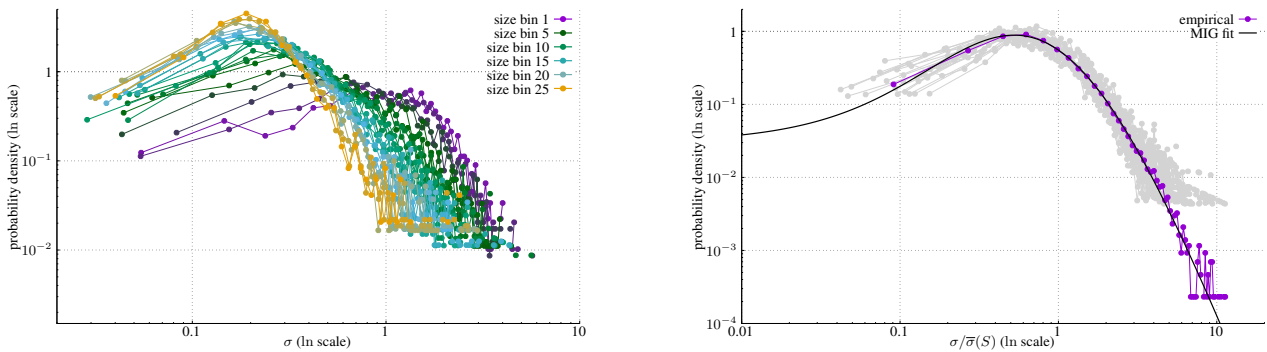

  \caption{Growth volatility distribution (volatility is computed as standard deviation)}
  \label{fig-app:pred1_sd}
  \begin{center}    
    \begin{minipage}[t]{0.47\linewidth}      
     \scalebox{0.6}{
      \input{./figures/fig_SD_binned.tex}
      }
    \end{minipage}
    \hfill
    \begin{minipage}[t]{0.47\linewidth}
     \scalebox{0.6}{
      \input{./figures/fig_SD_rescaled_binned.tex} 
      }
    \end{minipage}
  \end{center}
  {\scriptsize {\it Notes}. \underline{Left panel} reports on a double
    log scale and for 25 bins, defined in term of normalized size, the
    distribution of the growth rate volatility. \underline{Right
      panel} reports on a double log scale the distribution of the
    growth volatility rescaled by the average volatility observed in
    each bin together with an Inverse Gamma fit (solid line). The ML
    estimates of the scale, shape and location parameter are
    4.922(0.147), 4.668(0.089), and 0.344(0.010) respectively. In all
    the figures volatility is computed as the standard deviation. Data
    source: Compustat. \par}  
\end{figure}

\begin{figure}[H]
  \caption{Size and growth volatility (volatility is computed as
    standard deviation)}
   \label{fig-app:scaling_sd}
   \begin{center}
     \scalebox{1}{
       \input{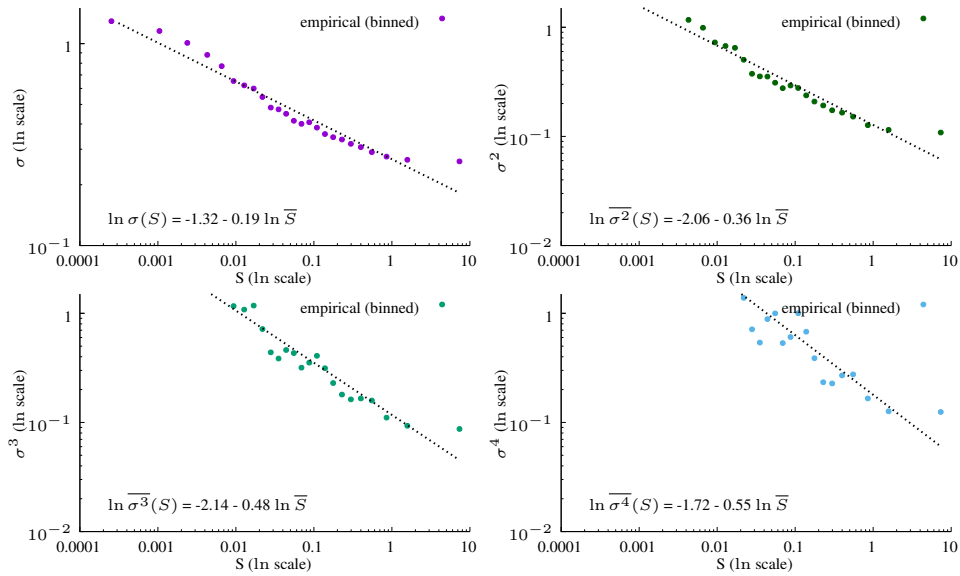}
     }
   \end{center}
   {\scriptsize {\it Notes}. The four panels report on a double log
     scale the binned relation between normalized size and the first
     four sample moments of growth volatility together with an OLS
     linear fit. In all panel the number of bins is set to $25$ and
     volatility is computed as the standard deviation. Data source:
     Compustat. \par}
 \end{figure}

\subsection{Sample with firms with 20 or more growth rates}
\label{app:20obs}

\noindent In this Appendix to alleviate concerns associated with
measurement errors in estimating growth volatilities we change the
sample and consider firms with $20$ or more growth rates only.

\medskip

\noindent {\bf Descriptive
  statistics}. Table~\ref{tab:sample-desc-20obs} reports statistics at
the firm-year level. The overall picture remains very similar to that
reported for the main sample.
\begin{table}[H]
  \centering 
  \caption{Descriptive statistics on size and growth rates (firms with 20 or more growth rates)} 
  \label{tab:sample-desc-20obs}
  \scalebox{0.6}{    
    \begin{threeparttable}
\begin{tabular}{@{\extracolsep{5pt}}lccccc} 
\\[-1.8ex]\hline 
\hline \\[-1.8ex] 
 & \multicolumn{1}{c}{N} & \multicolumn{1}{c}{Mean} & \multicolumn{1}{c}{St. Dev.} & \multicolumn{1}{c}{Min} & \multicolumn{1}{c}{Max} \\ 
\hline \\[-1.8ex] 
size $\tilde{S}$ & 1,103,590 & 379,856,536.000 & 2,155,729,951.000 & 1,000.000 & 152,079,000,000.000 \\ 
size $\tilde{S}$ deflated & 1,103,590 & 525,526,423.000 & 2,548,717,299.000 & 940.955 & 147,203,757,866.000 \\ 
size $S$ & 1,103,590 & 1.000 & 4.528 & 0.00000 & 245.280 \\ 
growth rate $\tilde{g}$ & 1,032,199 & 0.087 & 0.543 & $-$12.324 & 12.435 \\ 
growth rate $\tilde{g}$ deflated & 1,032,199 & 0.058 & 0.543 & $-$12.386 & 12.403 \\ 
growth rate $g$ & 1,032,199 & 0.037 & 0.545 & $-$12.123 & 12.313 \\ 
\hline \\[-1.8ex] 
\end{tabular} 
      \begin{tablenotes}
        \footnotesize
      \item {\it Notes}: Growth rates are computed as logarithmic
        differences. Data source: Compustat. \par
      \end{tablenotes}
    \end{threeparttable}
  }
\end{table}

Table~\ref{tab-app:sample-desc-20obs-f} reports statistics firm
level. In this sample the number of firms reduces to $15,788$ but, a
part of that, growth rates display statistical properties similar to
those observed in the main sample.
\begin{table}[H]
  \centering 
  \caption{Descriptive statistics on size and growth volatility (firms with 20 or more growth rates)} 
  \label{tab-app:sample-desc-20obs-f}
  \scalebox{0.6}{
    \begin{threeparttable}  
\begin{tabular}{@{\extracolsep{5pt}}lccccc} 
\\[-1.8ex]\hline 
\hline \\[-1.8ex] 
 & \multicolumn{1}{c}{N} & \multicolumn{1}{c}{Mean} & \multicolumn{1}{c}{St. Dev.} & \multicolumn{1}{c}{Min} & \multicolumn{1}{c}{Max} \\ 
\hline \\[-1.8ex] 
size $S$ & 15,788 & 0.573 & 2.502 & 0.00001 & 100.669 \\ 
num growth rates & 15,788 & 65.379 & 45.192 & 20 & 236 \\ 
growth volatility (mad) & 15,788 & 0.399 & 0.367 & 0.022 & 3.171 \\ 
growth volatility (sd) & 15,788 & 0.458 & 0.421 & 0.028 & 3.459 \\ 
\hline \\[-1.8ex] 
\end{tabular} 
      \begin{tablenotes}
        \footnotesize
      \item {\it Notes}: ``mad'' and ``sd'' represent the mean absolute
        deviation adjusted by the factor $\sqrt{\pi/2}$ and the
        standard deviation. Data source: Compustat. \par
      \end{tablenotes}
    \end{threeparttable}
  }
\end{table}

\noindent {\bf Growth volatility}. Figure~\ref{fig-app:pred1_20obs} and
Figure~\ref{fig-app:scaling_20obs} report the key empirical results on
growth volatility. The shape of its distribution, its scaling with
size and the quantitative fit of the Modified Inverse Gamma
distribution (MIG) remain qualitatively the same as in the main
sample.
\begin{figure}[H]
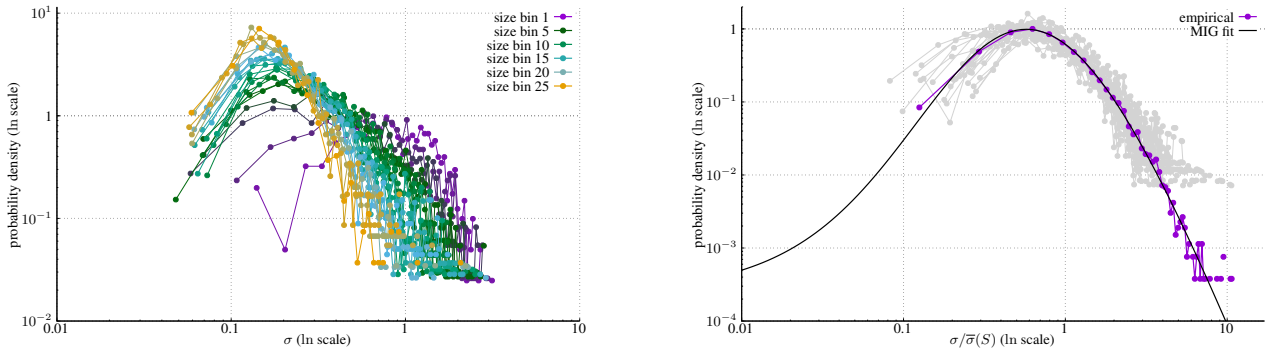

  \caption{Growth volatility distribution (firms with 20 or more growth rates)}
  \label{fig-app:pred1_20obs}
  \begin{center}    
    \begin{minipage}[t]{0.47\linewidth}      
        \scalebox{0.6}{
      \input{./figures/fig_MAD_binned_20obs.tex}
      }
    \end{minipage}
    \hfill
    \begin{minipage}[t]{0.47\linewidth}
     \scalebox{0.6}{
      \input{./figures/fig_MAD_rescaled_binned_20obs.tex} 
      }
    \end{minipage}
  \end{center}
  {\scriptsize {\it Notes}. \underline{Left panel} reports on a double
    log scale and for 25 bins, defined in term of normalized size, the
    distribution of the growth rate volatility. \underline{Right
      panel} reports on a double log scale the distribution of the
    growth volatility rescaled by the average volatility observed in
    each bin together with an Inverse Gamma fit (solid line). The ML
    estimates of the scale, shape and location parameter are
    4.381(0.151), 4.638(0.104), and 0.204(0.010) respectively. In all
    the figures volatility is computed as the mean absolute deviation
    multiplied by the factor $\sqrt{\pi/2}$. Data source:
    Compustat. \par}
\end{figure}
\begin{figure}[H]
  \caption{Size and growth volatility (firms with 20 or more growth rates)}
    \label{fig-app:scaling_20obs}
  \begin{center}
    \input{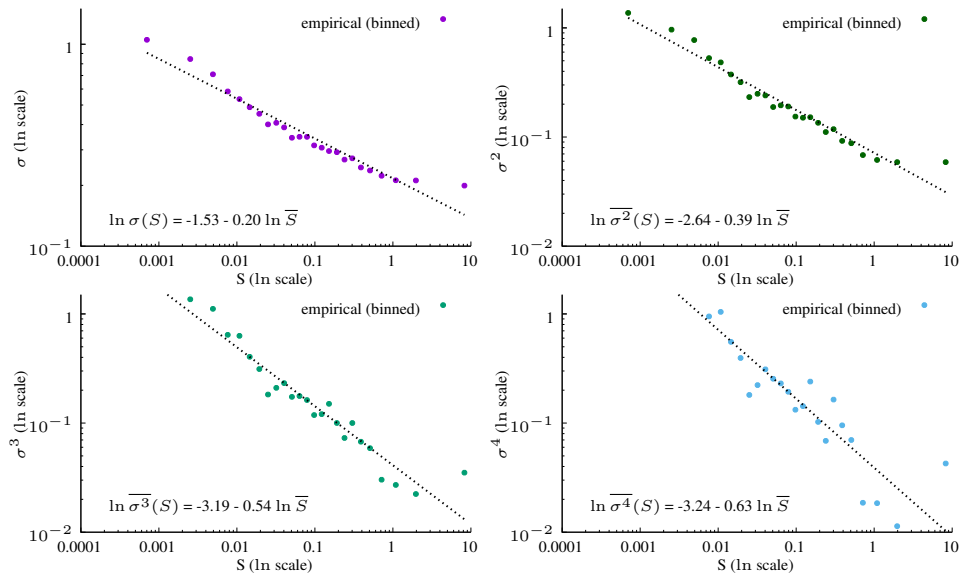} 
  \end{center}
  {\scriptsize {\it Notes}. The four panels report on a double log
    scale the binned relation between normalized size and the first
    four sample moments of growth volatility together with an OLS
    linear fit. In all panel the number of bins is set to $25$ and
    volatility is computed as the mean absolute deviation multiplied
    by the factor $\sqrt{\pi/2}$. Data source: Compustat. \par}
\end{figure}

\noindent {\bf GSE fit on rescaled growth rates}. Table~\ref{tab-app:gse_sym_20obs}
reports the GSE estimates for this sample. The main message discussed
in the paper looks qualitatively unaffected by this change in the
sample selection criterion.
\begin{table}[H]
  \centering 
  \caption{Symmetric GSE fit (firms with 20 or more growth rates)} 
  \label{tab-app:gse_sym_20obs}
  \scalebox{0.6}{
    \begin{threeparttable}
\begin{tabular}{@{\extracolsep{5pt}} l cccccc} 
\\[-1.8ex]\hline 
\hline \\[-1.8ex] 
 & $C$ & $u$ & $v$ & $w$ & $z$ & Prob. mass $[-w,w]$ \\ 
\hline \\[-1.8ex] 
\underline{{\it Whole sample}} &   &   &   &   &   &  \\ 
homogeneous rescaling & 0.75 (0.003) & 0.452 (0.002) & 0.012 (0.000) & 0.715 (0.004) & 0.413 (0.001) & 0.699 \\ 
heterogeneous rescaling & 0.503 (0.002) & 0.846 (0.004) & 0.024 (0.001) & 1.728 (0.013) & 0.417 (0.003) & 0.895 \\ 
 &   &   &   &   &   &  \\
\end{tabular} 
      \begin{tablenotes}
        \footnotesize
      \item {\it Notes}: estimates are obtained with a non lienar
        least square algorithm applied to kernel density estimates
        evaluated on a regular grid in the interval $[-8,8]$. This
        procedure is applied to homogeneously rescaled growth rates,
        and heterogeneously rescaled growth rates in different size
        bins and computed over different time windows.  Data source:
        Compustat. \par
      \end{tablenotes}
    \end{threeparttable}
    }
\end{table}

\subsection{Sample with firms whose fiscal year ends in December}
\label{app:endDec}

\noindent In this Appendix we explore if the presence in the sample of
firms with asynchronous fiscal years affects our results. To this aim
we change the sample and consider firms whose fiscal year ends in
December only.

\medskip

{\bf Descriptive statistics}. Table~\ref{tab-app:sample-desc-fyr-end-Dec}
reports statistics at the firm-year level. The overall picture remains
very similar to that reported for the main sample.

\begin{table}[H]
  \centering 
  \caption{Descriptive statistics on size and growth rates (firms whose fiscal year ends in December)} 
  \label{tab-app:sample-desc-fyr-end-Dec}
  \scalebox{0.6}{    
    \begin{threeparttable}
\begin{tabular}{@{\extracolsep{5pt}}lccccc} 
\\[-1.8ex]\hline 
\hline \\[-1.8ex] 
 & \multicolumn{1}{c}{N} & \multicolumn{1}{c}{Mean} & \multicolumn{1}{c}{St. Dev.} & \multicolumn{1}{c}{Min} & \multicolumn{1}{c}{Max} \\ 
\hline \\[-1.8ex] 
size $\tilde{S}$ & 717,545 & 431,790,843.000 & 2,135,254,603.000 & 1,000.000 & 132,872,000,000.000 \\ 
size $\tilde{S}$ deflated & 717,545 & 600,634,668.000 & 2,605,466,179.000 & 940.955 & 147,203,757,866.000 \\ 
size $S$ & 717,545 & 1.089 & 4.632 & 0.00000 & 215.495 \\ 
growth rate $\tilde{g}$ & 670,165 & 0.089 & 0.543 & $-$12.324 & 12.435 \\ 
growth rate $\tilde{g}$ deflated & 670,165 & 0.061 & 0.543 & $-$12.386 & 12.403 \\ 
growth rate $g$ & 668,705 & 0.047 & 0.543 & $-$12.224 & 12.328 \\ 
\hline \\[-1.8ex] 
\end{tabular} 
      \begin{tablenotes}
        \footnotesize
      \item {\it Notes}: Growth rates are computed as logarithmic
        differences. Data source: Compustat. \par
      \end{tablenotes}
    \end{threeparttable}
  }
\end{table}

Table~\ref{tab-app:sample-desc-fyr-end-Dec-f} reports statistics firm
level. In this sample the number of firms reduces to $10,974$ but, a
part of that, growth rates display statistical properties similar to
those observed in the main sample.

\begin{table}[H]
  \centering 
  \caption{Descriptive statistics on size and growth volatility (firms whose fiscal year ends in December)} 
  \label{tab-app:sample-desc-fyr-end-Dec-f}
  \scalebox{0.6}{
    \begin{threeparttable}  
\begin{tabular}{@{\extracolsep{5pt}}lccccc} 
\\[-1.8ex]\hline 
\hline \\[-1.8ex] 
 & \multicolumn{1}{c}{N} & \multicolumn{1}{c}{Mean} & \multicolumn{1}{c}{St. Dev.} & \multicolumn{1}{c}{Min} & \multicolumn{1}{c}{Max} \\ 
\hline \\[-1.8ex] 
size $S$ & 10,974 & 0.622 & 2.591 & 0.00001 & 94.560 \\ 
num growth rates & 10,974 & 60.935 & 45.060 & 2 & 233 \\ 
growth volatility (mad) & 10,974 & 0.395 & 0.382 & 0.000 & 5.380 \\ 
growth volatility (sd) & 10,974 & 0.453 & 0.433 & 0.000 & 5.168 \\ 
\hline \\[-1.8ex] 
\end{tabular} 
      \begin{tablenotes}
        \footnotesize
      \item {\it Notes}: ``mad'' and ``sd'' represent the mean absolute
        deviation adjusted by the factor $\sqrt{\pi/2}$ and the
        standard deviation. Data source: Compustat. \par
      \end{tablenotes}
    \end{threeparttable}
  }
\end{table}

\noindent {\bf Growth volatility}. Figure~\ref{fig-app:pred1_fyr_end_Dec}
and Figure~\ref{fig-app:scaling_fyr_end_Dec} report the key empirical
results on growth volatility. The shape of its distribution, its
scaling with size and the quantitative fit of the Modified Inverse
Gamma distribution (MIG) remain qualitatively the same as in the main
sample.

\begin{figure}[H]
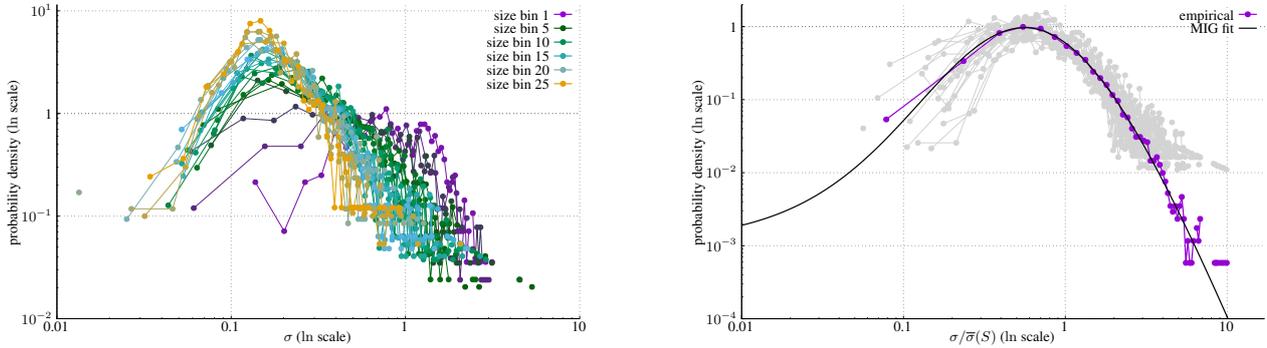

  \caption{Growth volatility distribution (firms whose fiscal year ends in December)}
  \label{fig-app:pred1_fyr_end_Dec}
  \begin{center}    
    \begin{minipage}[t]{0.47\linewidth}  
     \scalebox{0.6}{
          \input{./figures/fig_MAD_binned_fyr_end_Dec.tex}
          }
    \end{minipage}
    \hfill
    \begin{minipage}[t]{0.47\linewidth}
     \scalebox{0.6}{
      \input{./figures/fig_MAD_rescaled_binned_fyr_end_Dec.tex} 
      }
    \end{minipage}
  \end{center}
  {\scriptsize {\it Notes}. \underline{Left panel} reports on a double
    log scale and for 25 bins, defined in term of normalized size, the
    distribution of the growth rate volatility. \underline{Right
      panel} reports on a double log scale the distribution of the
    growth volatility rescaled by the average volatility observed in
    each bin together with an Inverse Gamma fit (solid line). The ML
    estimates of the scale, shape and location parameter are    
    4.339(0.174), 4.540(0.119), and 0.225(0.012) respectively. In all
    the figures volatility is computed as the mean absolute deviation
    multiplied by the factor $\sqrt{\pi/2}$. Data source:
    Compustat. \par}
\end{figure}
\begin{figure}[H]
  \caption{Size and growth volatility (only firms with end of fiscal year in December)}
  \label{fig-app:scaling_fyr_end_Dec}
  \begin{center}
    \input{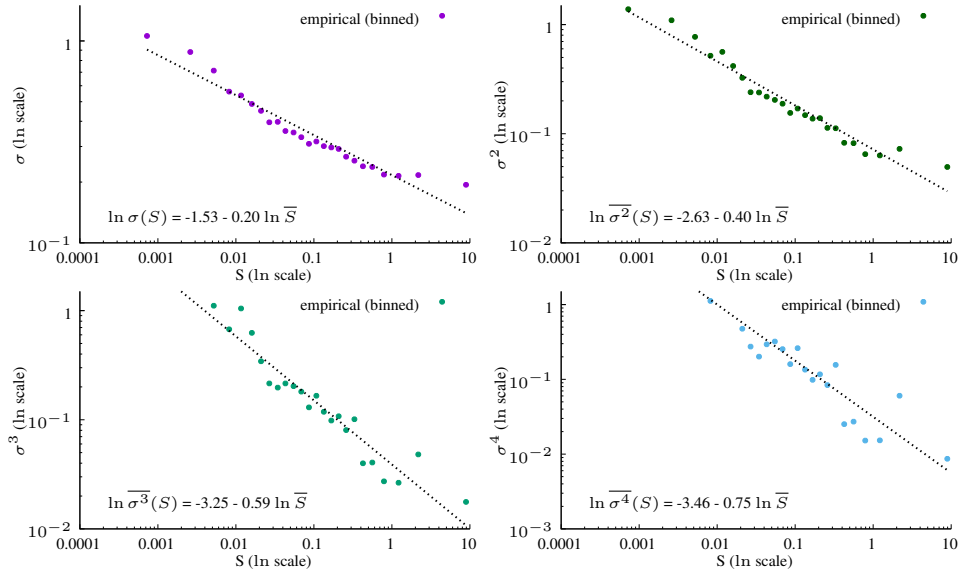} 
  \end{center}
  {\scriptsize {\it Notes}. The four panels report on a double log
    scale the binned relation between normalized size and the first
    four sample moments of growth volatility together with an OLS
    linear fit. In all panel the number of bins is set to $25$ and
    volatility is computed as the mean absolute deviation multiplied
    by the factor $\sqrt{\pi/2}$. Data source: Compustat. \par}
\end{figure}

\noindent {\bf GSE fit on rescaled growth
  rates}. Table~\ref{tab-app:gse_sym_fyr_end_Dec} reports the GSE estimates
for this sample. The main message discussed in the paper looks
qualitatively unaffected by this change in the sample selection
criterion.

\begin{table}[H]
  \centering 
  \caption{Symmetric GSE fit} 
  \label{tab-app:gse_sym_fyr_end_Dec}
  \scalebox{0.6}{
    \begin{threeparttable}
\begin{tabular}{@{\extracolsep{5pt}} l cccccc} 
\\[-1.8ex]\hline 
\hline \\[-1.8ex] 
 & $C$ & $u$ & $v$ & $w$ & $z$ & Prob. mass $[-w,w]$ \\ 
\hline \\[-1.8ex] 
\underline{{\it Whole sample}} &   &   &   &   &   &  \\ 
homogeneous rescaling & 0.768 (0.004) & 0.431 (0.002) & 0.015 (0.000) & 0.67 (0.005) & 0.416 (0.001) & 0.685 \\ 
heterogeneous rescaling & 0.512 (0.003) & 0.819 (0.005) & 0.027 (0.001) & 1.629 (0.017) & 0.43 (0.003) & 0.882 \\ 
 &   &   &   &   &   &  \\
\end{tabular} 
      \begin{tablenotes}
        \footnotesize
      \item {\it Notes}: estimates are obtained with a non lienar
        least square algorithm applied to kernel density estimates
        evaluated on a regular grid in the interval $[-8,8]$. This
        procedure is applied to homogeneously rescaled growth rates,
        and heterogeneously rescaled growth rates in different size
        bins and computed over different time windows.  Data source:
        Compustat. \par
      \end{tablenotes}
    \end{threeparttable}
    }
\end{table}

\end{document}